\documentclass{article}
\pdfoutput=1
\usepackage{structure}
\newcommand{\fullversion}{Print full version.}

\usepackage{url}

\begin{document}
\maketitle
\begin{abstract}
    Ben Reichardt showed in a series of results that the general adversary bound of a function characterizes its quantum query complexity. This survey seeks to aggregate the background and definitions necessary to understand the proof. Notable among these are the lower bound proof, span programs, witness size, and semi-definite programs. These definitions, in addition to examples and detailed expositions, serve to give the reader a better intuition of the graph-theoretic nature of the upper bound. We also include an applications of this result to lower bounds on DeMorgan formula size. 
\end{abstract}

\renewcommand{\baselinestretch}{0.75}\normalsize
\tableofcontents
\renewcommand{\baselinestretch}{1.0}\normalsize

\section{Introduction}
Given a function $f$, the quantum query complexity of $f$, denoted $Q_2(f)$, is the number of quantum oracle queries necessary to evaluate $f$. It is typically used as a lower bound on the complexity of a quantum algorithm: the amount of computation allowed between queries is unbounded, so the analysis can be much simpler. The \emph{polynomial method} ~\cite{DBLP:journals/jacm/BealsBCMW01} and the \emph{adversary bound method}~\cite{DBLP:journals/jcss/Ambainis02,DBLP:conf/coco/BarnumSS03} are common techniques used to show lower bounds on quantum query complexity. However, these techniques are currently incomparable; on the $n$-input collision problem, the adversary method only achieves an $O(1)$ lower bound while the polynomial method achieves the optimal $\Theta(n^{1/3})$ bound while on Ambainis' total function $f^{k}$ on $4^{k}$ bits, the polynomial method achieves at most a $2^k$ lower bound which is strictly weaker than the adversarial bound of $2.5^{k}$ \cite{DBLP:journals/jcss/Ambainis06}. 

The adversary bound was originally proposed by Ambainis~\cite{DBLP:journals/jcss/Ambainis02}. Given a boolean function $f$, the adversary bound of $f$, denoted $\advnonnegative(f)$, captures the intuition that, in order to compute $f$, one must be able to distinguish between any two inputs $w$ and $x$ where $f(w) \not= f(x)$. Specifically, if $\ket{\phi_w}$ and $\ket{\phi_x}$ are the final state of a quantum query algorithm after running with inputs $w$ and $x$ respectively, then $\ket{\phi_w}$ and $\ket{\phi_x}$ must be far apart in our measurement basis if $f(w) \not= f(x)$.\footnote{We assume familiarity with bra-ket notation, which is used here and elsewhere in the survey.} There are several equivalent formulations of the bound. This survey uses the spectral norm formulation of Barnum, Saks, and Szegedy~\cite{DBLP:conf/coco/BarnumSS03}. 

\begin{definition}
	An \emph{adversary matrix} for $f: \{0, 1\}^n \rightarrow \{0, 1\}$ is a $2^n$-by-$2^n$ Hermitian matrix $\Gamma$ where $\bra{x}\Gamma\ket{y} = 0$ whenever $f(x) = f(y)$.
\end{definition}

\begin{definition}
    The matrix  $D_i$ is the $2^n$-by-$2^n$ matrix where $\bra{x}D_i\ket{y} = 0$ if $x_i = y_i$ and $\bra{x}D_i\ket{y} = 1$ if $x_i \not= y_i$.
\end{definition}

\begin{definition}
	The \emph{adversary bound} on a function $f: \{0, 1\}^n \rightarrow \{0, 1\}$ is
	$$\advnonnegative(f) = \max_{\substack{\Gamma \geq 0\\\Gamma \neq 0}} \frac{\left\lVert \Gamma \right\rVert}{\max_i \left\lVert \Gamma \circ D_i \right\rVert}$$
	where $\Gamma$ is an adversary matrix for $f$.
\end{definition}

In the above definition, $\Gamma \geq 0$ indicates that all entries of $\Gamma$ are non-negative, and the operator $\circ$ denotes entry-wise product. As mentioned above, the adversary bound lower-bounds quantum query complexity.

\begin{theorem}[\cite{DBLP:conf/coco/BarnumSS03}]
	$\advnonnegative(f) = \Omega(Q_2(f))$.
\end{theorem}

Furthermore, Laplante, Lee, and Szegedy show that the adversary bound of a function is a lower bound on the square root of the function's De Morgan formula size.

\begin{theorem}[\cite{DBLP:journals/cc/LaplanteLS06}]
	$\advnonnegative(f) \leq \sqrt{\mathcal{L}_f}$
\end{theorem}

H{\o}yer, Lee, and \v{S}palek~\cite{DBLP:conf/stoc/HoyerLS07} removed the non-negativity requirement from the adversary bound and showed that it remained a lower bound on quantum query complexity and formula size. In fact this generalization only strengthened the lower bound -- it is a tight bound on quantum query complexity, although this was not shown until later by Reichardt~\cite{DBLP:conf/soda/Reichardt11a}.

\begin{definition}
	The \emph{general adversary bound} on a function $f: \{0, 1\}^n \rightarrow \{0, 1\}$ is
	$$\adv(f) = \max_{\Gamma \neq 0} \frac{\left\lVert \Gamma \right\rVert}{\max_i \left\lVert \Gamma \circ D_i \right\rVert}$$
	where $\Gamma$ is an adversary matrix for $f$.
\end{definition}

Unlike the matching lower bound, which has a relatively simple proof, the upper bound is proved using deceptively simple algorithms with complicated analyses. Key to this analysis is the span program model of computation \cite{DBLP:conf/coco/KarchmerW93}. Reichardt's main result is the following:
\begin{theorem}{\textup{(General Adversary Bound Characterize Quantum Query Complexity.)}}
    \label{thm:adv-characterize-q}
    For any $n$-ary boolean function $f$,
    \[Q_2(f) = \Theta(\adv(f)).\]
\end{theorem}

Perhaps surprisingly, this result about quantum computing has been used recently in entirely classical settings. From the work of Ambainis et al.\ \cite{DBLP:journals/siamcomp/AmbainisCRSZ10}, any function with de Morgan formula size $\ell$ has a quantum query algorithm which makes at most $O(\sqrt{\ell})$ queries. In combination with the polynomial method, this implies the existence of a polynomial $p$ with degree $O(\sqrt{\ell})$ such that $p(x)$ approximates $f$ up to a constant factor. More recently, Tal used the result to show a $\tilde{\Omega}(n^2)$ lower bound for the bipartite formula size of the Inner-Product function \cite{DBLP:conf/stoc/Tal17}.

\section{Preliminaries}
Let $f$ be an $n$-ary boolean function. Then $F_0 = f^{-1}(0)$ and $F_1 = f^{-1}(1)$ are sets of strings which evaluate to $0$ and $1$ on $f$ respectively.

We assume basic familiarity with quantum computation and bra-ket notation. Given a vector $\ket{v}$, let $\norm{\ket{v}}$ represent the $\ell_2$-norm of $\ket{v}$. Given a matrix $M$, let $\norm{M}$ represent the \emph{spectral norm} of the matrix, defined as $\max_{\ket{u}} \norm{M\ket{u}}$ where the maximum is over unit vectors $\ket{u}$. In this survey, we use the fact that $\norm{M}$ is the largest singular value of $M$. For two matrices $A$ and $B$, their entry-wise product is denoted $A\circ B$ and their entry-wise inner product is denoted $\langle A, B \rangle$. The \emph{trace} of a square matrix is the sum of its eigenvalues (including multiplicities). In this survey we make use of the fact that $|\langle M, B\rangle| = \Tr(M^*B)$ when $M$ and $B$ are square matrices with the same dimension. 

Let the \emph{trace norm} of a matrix $M$ be $\norm{M}_{\textup{Tr}} = \max_B |\langle M, B \rangle | / \norm{B}$ i.e.\ the maximization of the trace of $M$ over all complex matrices with the same dimensions as $M$. Another standard definition of the trace norm is $\norm{M}_{\textup{Tr}} = \Tr(\sqrt{M^* M})$. These definitions are equivalent, which can be proved by observing that the spectral norm is the Schatten $\infty$-norm and the trace norm is the Schatten 1-norm, and using the fact that, for $\frac{1}{p} + \frac{1}{q} = 1$, the Schatten $p$-norm of $M$ is $\max_B |\langle M, B \rangle |$ divided by the Schatten $q$-norm of $B$. The Frobenius norm of a matrix $M$, denoted $\norm{M}_F$, is $\sqrt{\sum_{x, y} \left(\bra{x}M\ket{y}^2\right)}$. We use the fact that the Frobenius norm is the Schatten 2-norm.

For a matrix $A$, we say that $A \in \ltrans(U, V)$ if $A$ is a linear transformation from vectors in $\CC^{U}$ to vectors in $\CC^{V}$. In this case $A$ has $|U|$ columns and $|V|$ rows. Let $\ltrans(U) = \ltrans(U, U)$. $\id_k$ is the $k \times k$ identity matrix. When the dimensions are clear from context, we omit the subscript. For $i \in U$ and $j \in V$, we will use $\ket{i}$ and $\ket{j}$ to denote the indicator vectors for the relevant column and row of $A$ respectively. In particular, $\bra{i}A\ket{j}$ is the entry of $A$ in row $i$ and column $j$.

Readers will also require some familiarity with positive semi-definite matrix (PSD) and semi-definite programs (SDP). If $X$ is PSD, we write $X \succeq 0$. When we write $X \succeq Y$, we mean $X-Y \succeq 0$. 

\section{The General Adversary Bound}

\ifdefined\original
The adversary bound was originally proposed by Ambainis~\cite{DBLP:journals/jcss/Ambainis02} and has many different equivalent formulations. Here we will use a formulation by Barnum, Saks, and Szegedy that uses the spectral norm~\cite{DBLP:conf/coco/BarnumSS03}. 

\begin{definition}
	An \emph{adversary matrix} for $f: \{0, 1\}^n \rightarrow \{0, 1\}$ is a $2^n$-by-$2^n$ Hermitian matrix $\Gamma$ where $\bra{x}\Gamma\ket{y} = 0$ whenever $f(x) = f(y)$.
\end{definition}

\begin{definition}
    The matrix  $D_i$ is the $2^n$-by-$2^n$ matrix where $\bra{x}D_i\ket{y} = 0$ if $x_i = y_i$ and $\bra{x}D_i\ket{y} = 1$ if $x_i \not= y_i$.
\end{definition}

\begin{definition}
	The \emph{adversary bound} on a function $f: \{0, 1\}^n \rightarrow \{0, 1\}$ is
	$$\advnonnegative(f) = \max_{\substack{\Gamma \geq 0\\\Gamma \neq 0}} \frac{\left\lVert \Gamma \right\rVert}{\max_i \left\lVert \Gamma \circ D_i \right\rVert}$$
	where $\Gamma$ is an adversary matrix for $f$.
\end{definition}

\begin{theorem}[\cite{DBLP:conf/coco/BarnumSS03}]
	$\advnonnegative(f) = \Omega(Q_2(f))$.
\end{theorem}

Furthermore, Laplante, Lee, and Szegedy show that the adversary bound of a function is a lower bound on the square root of the function's De Morgan formula size.

\begin{theorem}[\cite{DBLP:journals/cc/LaplanteLS06}]
	$\advnonnegative(f) \leq \sqrt{\mathcal{L}_f}$
\end{theorem}

H{\o}yer, Lee, and \v{S}palek~\cite{DBLP:conf/stoc/HoyerLS07} removed the non-negativity requirement from the adversary bound and showed that it remained a lower bound on quantum query complexity and formula size. In fact this generalization only strengthened the lower bound -- indeed, it is a tight bound on quantum query complexity, although this was not shown until later by Reichardt~\cite{DBLP:conf/soda/Reichardt11a}.

\begin{definition}
	The \emph{general adversary bound} on a function $f: \{0, 1\}^n \rightarrow \{0, 1\}$ is
	$$\adv(f) = \max_{\Gamma \neq 0} \frac{\left\lVert \Gamma \right\rVert}{\max_i \left\lVert \Gamma \circ D_i \right\rVert}$$
	where $\Gamma$ is an adversary matrix for $f$.
\end{definition}
\fi

In this section we show the following properties of the general adversary bound. 
\begin{theorem}[\cite{DBLP:conf/stoc/HoyerLS07}] \label{thm:adv-q2-lb}
	$\adv(f) = \Omega(Q_2(f))$.
\end{theorem}

\begin{theorem}[\cite{DBLP:conf/stoc/HoyerLS07}] \label{thm:adv-l-lb}
	$\adv(f) \leq \sqrt{\mathcal{L}_f}$
\end{theorem}

\ifdefined\fullversion
The corresponding upper bound, $\adv(f) = O(Q_2(f))$, will be left to a later section.

\subsection{\texorpdfstring{$\adv$}{Adv+-} is a Lower Bound for Quantum Query Complexity}
\else %
We will present the corresponding quantum query complexity upper bound ($\adv(f) = O(Q_2(f))$) in a later section, and defer the proof of Theorem~\ref{thm:adv-l-lb} to Appendix~\ref{sec:adv-l-lb}. 

\fi %

\newcommand{\ind}{\mathsf{IND}}
\newcommand{\W}[1]{\mathsf{W}^{(#1)}}
\newcommand{\M}[1]{\mathsf{M}^{(#1)}}
\newcommand{\Wstar}[1]{\mathsf{W}_\star^{(#1)}}

Consider a quantum query algorithm that computes $f$ in $T$ steps with error at most $1/3$. Without loss of generality, the quantum query algorithm is of the form $U_T V_{\ind} U_{T-1} V_{\ind} \cdots U_1 V_{\ind} U_0$ where each $U_t$ is a unitary that does not depend on the input $x$ and $V_{\ind}$ is the standard phase oracle unitary on the \emph{index function} $\ind$:
$$V_{\ind}\ket{i} = (-1)^{x_i}\ket{i}$$
It will be helpful to divide the state of this quantum query algorithm into three sets of qubits: (1) the \emph{input set} $I$ holds the input $x$ and remains unchanged throughout the execution of the algorithm, (2) the \emph{query set} $Q$ that is used by each $V_{\ind}$ to specify a coordinate of $x$, and (3) a \emph{workspace set} $W$ that can be acted upon arbitrarily. The qubits in $Q$ and $W$ are measured at the end of the algorithm to obtain an output in $\{0, 1\}$. These measurements can be viewed as orthogonal projectors $\Pi_0, \Pi_1$. Let the combined state of $Q$ and $W$ on input $x$ at step $t$ be $\ket{\psi^t_x}$. Further define the matrix $\Psi^{t}$ to be the $2^{n}\times (|Q|+|W|)$ matrix with $\ket{\psi^{t}_x}$ as rows. Then the probability that we will measure outcome $b$ on input $x$ is $\norm{\Pi_b \ket{\psi^T_x}}^2$. Note that for all $x$ we have $\norm{\Pi_{f(x)} \ket{\psi^T_x}}^2 \geq 2/3$. Three other important properties of the projectors are that $\Pi_0 + \Pi_1 = \mathbb{I}$ (the projectors are complete), $\Pi_b^2 = \Pi_b$ (performing a projection twice has no more effect than applying it once), and $\Pi_0 \Pi_1 = \Pi_1 \Pi_0 = 0$ (the projections are orthogonal).

The main observation that H{\o}yer, Lee, and \v{S}palek~\cite{DBLP:conf/stoc/HoyerLS07} use in their proof of Theorem~\ref{thm:adv-q2-lb} is that the combined state of $Q$ and $W$ must be very different when the algorithm is run on $x$ compared to when it is run on $y$ if $f(x) \not= f(y)$: otherwise, any measurement would be unable to distinguish these states with high enough fidelity. We present their argument here.

\begin{proof}[Proof of Theorem~\ref{thm:adv-q2-lb}]
	Let $\Gamma$ be an adversary matrix. Note that $\norm{\Gamma}$ is the largest absolute value of any eigenvalue of $\Gamma$, as $\Gamma$ is Hermitian. Assume that $\norm{\Gamma} = \lambda_1$ where $\lambda_1$ is the largest eigenvalue of $\Gamma$: this can be done without loss of generality by replacing $\Gamma$ with $(-1)\Gamma$, which does not affect the value of $\norm{\Gamma}$. Let $\ket{\delta}$ be the unit eigenvector corresponding to $\lambda_1$.
	
	Consider running our quantum query algorithm for $f$ with an input in a superposition defined by $\ket{\delta}$: the state of the input qubits will be $\sum_{x \in \{0, 1\}^n} \ip{x}{\delta}\ket{x}$. Then the state of $Q$ and $W$ at step $t$ will be $\sum_{x \in \{0, 1\}^n} \ip{x}{\delta}\ket{\psi_x^t}$. Let $\W{t}$ be the $2^n$-by-$2^n$ density matrix defined by $\bra{x}\W{t}\ket{y} = \ip{x}{\delta}\ip{\delta}{y}\ip{\psi_y^t}{\psi_x^t}$. Equivalently, $\W{t} = \ket{\delta}\bra{\delta} \circ \Psi^{t}\left(\Psi^{t}\right)^{*}$.\footnote{To see that $\W{t}$ is indeed a density matrix, note that it is the Gram matrix of $\{ \ip{x}{\delta} \ket{\psi_x^t} \}$.} We measure the progress of the algorithm by comparing $\W{t}$ to $\Gamma$. Define the progress measure $\M{t} = \langle \Gamma, \W{t} \rangle$. To prove the lower bound, it suffices to show that this progress measure changes by an amount bounded above by $2 \max_i \norm{\Gamma \circ D_i}$ at each step of the algorithm, but must change by at least a constant multiple of $\norm{\Gamma}$ over the course of the entire algorithm. The following three claims show this.
	
	\begin{claim}
		$\M{0} = \norm{\Gamma}$
	\end{claim}
	\begin{proof}
		Before any executions of the phase oracle, the state cannot depend on the input: for all $x$ and $y$, $\ket{\psi_x^0} = \ket{\psi_y^0}$, and so $\W{0} = \ket{\delta}\bra{\delta}$. Then $\M{0} = \langle \Gamma, \ket{\delta}\bra{\delta} \rangle = \textup{Tr}(\Gamma^* \ket{\delta}\bra{\delta}) = \textup{Tr}(\lambda_1 \ket{\delta}\bra{\delta}) = \lambda_1 \cdot 1 = \norm{\Gamma}$.
	\end{proof}

\ifdefined\fullversion
	\begin{claim}
		$\M{T} \leq (\frac{2}{3} \sqrt{2}) \norm{\Gamma}$
	\end{claim}
	\begin{proof}
		First note that $\Gamma = \Gamma \circ F$ where $F$ is the 0/1 adversary matrix:
\[\bra{x}F\ket{y} = \begin{cases}0 & f(x) = f(y)\\1 & f(x) \not= f(y)\end{cases}\]
Thus $\M{T} = \langle \Gamma \circ F, \W{T} \rangle = \langle \Gamma, F \circ \W{T} \rangle$. By the definition of the trace norm, this gives us $\M{T} \leq \norm{\Gamma} \norm{F \circ \W{T}}_{\textup{Tr}}$. To prove the claim we simply need to upper-bound $\norm{F \circ \W{T}}_{\textup{Tr}}$.

Let $X_0$ (respectively $X_1$) be the $2^n \times 2^n$ matrix where the $x$th row is (the conjugate of) $\Pi_{f(x)} \delta_x \ket{\psi_x^T}$ (respectively $\Pi_{1-f(x)} \delta_x \ket{\psi_x^T}$). Intuitively, $X_0$ is the projection onto the correct answers and $X_1$ is the projection onto the incorrect answers.

Observe that $F \circ \W{T} = X_0 X_1^* + X_1 X_0^*$:
$$
\bra{x}(X_0 X_1^* + X_1 X_0^*)\ket{y} = \ip{x}{\delta}\ip{\delta}{y}\left(\bra{\psi_y^T}\Pi_{1-f(y)}\Pi_{f(x)}\ket{\psi_x^T} + \bra{\psi_y^T}\Pi_{f(y)}\Pi_{1 - f(x)}\ket{\psi_x^T} \right).
$$
If $f(x) = f(y)$, then the expression on the right is 0, as $\Pi_0 \Pi_1 = 0$. Otherwise, $\Pi_b \Pi_b = \Pi_b$ and $\Pi_0 + \Pi_1 = \mathbb{I}$, so the expression on the right is $\ip{x}{\delta} \ip{\delta}{y} \ip{\psi_y^T}{\psi_x^T} = \bra{x}\W{T}\ket{y}$.

We now need to upper-bound $\norm{X_0 X_1^* + X_1 X_0^*}_{\textup{Tr}}$.
\begin{align*}
\norm{X_0 X_1^* + X_1 X_0^*}_{\textup{Tr}} & \leq \norm{X_0 X_1^*}_{\textup{Tr}} + \norm{X_1 X_0^*}_{\textup{Tr}} & \mbox{(by the triangle inequality)}\\
& \leq \norm{X_0}_F \norm{X_1^*}_F + \norm{X_1}_F \norm{X_0^*}_F  & \mbox{(by H{\"o}lder's Inequality)}\\
& = 2 \norm{X_0}_F \norm{X_1}_F
\end{align*}
H\"{o}lder's inequality applies here because the trace norm is the Schatten 1-norm and the Frobenius norm is the Schatten 2-norm. We upper-bound this final expression by noting the following two facts:
$$ \norm{X_0}_F^2 + \norm{X_1}_F^2 = \sum_{x \in \{0, 1\}^n} \left| \ip{x}{\delta} \right|^2 \left( \norm{\Pi_{f(x)} \ket{\psi_x^T}}^2 + \norm{\mathbb{I}-\Pi_{f(x)} \ket{\psi_x^T}}^2 \right) = \norm{\delta}^2 = 1 $$
$$ \norm{X_0}_F^2 = \sum_{x \in \{0, 1\}^n} \left| \ip{x}{\delta} \right|^2 \norm{\Pi_{f(x)} \ket{\psi_x^T}}^2 \geq \frac{2}{3} \norm{\delta}^2 = \frac{2}{3} $$
Therefore, $2 \norm{X_0}_F \norm{X_1}_F$ is maximized at $2 \sqrt{2/3} \sqrt{1/3} = \frac{2}{3} \sqrt{2}$.
	\end{proof}
	
	From the first two claims, we know that $\M{0} - \M{T} \geq (1 - \frac{2}{3} \sqrt{2}) \norm{\Gamma}$. The last step in the proof is to give an upper bound on $\M{t} - \M{t+1}$ for all $t$.
	
	\begin{claim}
		$\M{t} - \M{t+1} \leq 2 \max_{i} \norm{\Gamma \circ D_i}$
	\end{claim}
	\begin{proof}
		To help us prove this claim, we will define a new density matrix $\Wstar{t}$ that is similar to $\W{t}$. Whereas $\W{t}$ is indexed by the basis states of the input qubits $I$ and has entries defined by the state of the query qubits $Q$ and the workspace qubits $W$, $\Wstar{t}$ will be indexed by $I$ and $Q$ and have entries defined by the state of $W$.
$$\bra{x, i} \Wstar{t} \ket{y, j} = \ip{x}{\delta}\ip{\delta}{y}\ip{\psi_y^t}{j}\ip{i}{\psi_x^t}$$
Here, $i$ and $j$ are basis states of $Q$. Note that $W^{(t)}_{*}$ is a density matrix since it has trace one, and is positive semi-definite since $W^{(t)}_*$ is also a Gram matrix.

Let $G$ and $D$ be the following block-diagonal $(n\cdot2^n)\times(n\cdot2^n)$ matrices:
$$G = \Gamma \otimes \mathbb{I}_n \qquad D = \sum_{i=1}^n D_i \otimes |i\rangle \langle i|$$
Note that $\M{t} = \langle \Gamma, \W{t} \rangle = \langle G, \Wstar{t} \rangle$. We would like to give $\M{t+1}$ in terms of $\Wstar{t}$ as well, and so we analyze the effect of a single step of the quantum query algorithm on this matrix. Since the unitary $U_{t+1}$ does not depend on the input qubits, we can ignore it for the purposes of our progress measure: $\bra{\psi_y^t}U_{t+1}^*U_{t+1}\ket{\psi_x^t} = \ip{\psi_y^t}{\psi_x^t}$, and so $\W{t}$ does not change after the application of the unitary. This means that $\M{t+1} = \langle G, \Wstar{t+1} \rangle = \langle G, V_{\ind} \Wstar{t} V_{\ind}^* \rangle$.
$$	\M{t} - \M{t+1} = \langle G, \Wstar{t} \rangle - \langle G, V_{\ind} \Wstar{t} V_{\ind}^* \rangle
= \langle G, \Wstar{t} - V_{\ind} \Wstar{t} V_{\ind}^* \rangle
= \langle G, (\Wstar{t} - V_{\ind} \Wstar{t} V_{\ind}^*) \circ D \rangle $$
This last equality is true because $\bra{x, i} (\Wstar{t} - V_{\ind} \Wstar{t} V_{\ind}^*) \ket{y, i} = (1 - (-1)^{x_i + y_i}) \bra{x, i}\Wstar{t}\ket{y, i}$ (which is 0 when $x_i = y_i$) and $G$ is block-diagonal ($\bra{x, i}G\ket{y, j} = 0$ if $i \not= j$).
\begin{align*}
	\M{t} - \M{t+1} &= \langle G, (\Wstar{t} - V_{\ind} \Wstar{t} V_{\ind}^*) \circ D \rangle &\\
	&= \langle G \circ D, (\Wstar{t} - V_{\ind} \Wstar{t} V_{\ind}^*) \rangle&\\
	&\leq \norm{G \circ D} \cdot \norm{\Wstar{t} - V_{\ind} \Wstar{t} V_{\ind}^*}_{\textup{Tr}}&\mbox{(by the definition of the trace norm)}\\
	&\leq \norm{G \circ D} \cdot \left( \norm{\Wstar{t}}_{\textup{Tr}} + \norm{V_{\ind} \Wstar{t} V_{\ind}^*}_{\textup{Tr}} \right) &\mbox{(by the triangle inequality)}\\
	&=2 \norm{G \circ D} \norm{\Wstar{t}} &\mbox{($V_{\ind}$ is unitary)}\\
	&=2 \norm{G \circ D} &\mbox{($\Wstar{t}$ is a density matrix)}\\
	&=2 \max_i \norm{\Gamma \circ D_i} &
\end{align*}
In the above we used the following facts: the trace norm is invariant under conjugation with a unitary, the trace norm of a density matrix is one, and the spectral norm of a block-diagonal matrix is the maximum of the spectral norms of the blocks.
	\end{proof}
\else %
	\begin{claim}
	    \label{clm:clm2}
		$\M{T} \leq (\frac{2}{3} \sqrt{2}) \norm{\Gamma}$
	\end{claim}
	\begin{claim}
	    \label{clm:clm3}
		$\M{t} - \M{t+1} \leq 2 \max_{i} \norm{\Gamma \circ D_i}$
	\end{claim}
	
	The proof of Claim~\ref{clm:clm2} and Claim~\ref{clm:clm2} are deferred to Appendix~\ref{sec:q2-lb}.
\fi %

	Putting it all together, we have that a quantum query algorithm for $f$ requires at least $\frac{1 - \frac{2}{3}\sqrt{2}}{2} \adv{f} = \Omega(\adv{f})$ rounds.
\end{proof}

\ifdefined\fullversion
\subsection{\texorpdfstring{$\adv$}{Adv+-} is a Lower Bound for the Square Root of Formula Size}
The lower bound on $\sqrt{\mathcal{L}(f)}$ using $\adv(f)$ makes use of the \emph{Karchmer-Wigderson game on $f$}.

\begin{definition}[\cite{DBLP:journals/siamdm/KarchmerW90}]Given a Boolean function $f$, the \emph{Karchmer-Wigderson game on $f$} $(\kw(f))$ is a two-player communication game in which one party receives an input $x \in f^{-1}(0)$, one party receives an input $y \in f^{-1}(1)$, and the parties must collectively determine some coordinate $i$ on which $x_i \not= y_i$. 
\end{definition}

A useful fact is that the minimum number of leaves in a De Morgan formula that computes a function $f$ -- denoted $\mathcal{L}(f)$ -- is exactly the minimum number of leaves in a communication protocol that successfully solves $\kw(f)$ -- denoted $C^P(\kw(f))$.
	
\begin{theorem}[\cite{DBLP:journals/siamdm/KarchmerW90}]
	$\mathcal{L}(f) = C^P(\kw(f))$
\end{theorem}

We give a brief sketch of the proof, noting that we only need one direction for the lower bound in this section.

\begin{proof}[Proof (Sketch)]
	Given a formula for $f$, we can use induction on the depth of the formula to produce a communication protocol for $\kw(f)$. 
	
	If the formula is a single leaf, then that leaf must be labelled with some literal. Then, since the formula evaluates to false on $x$ and true on $y$, the players know that they differ on the leaf's literal and therefore no communication is required (and so the communication protocol for $\kw(f)$ is also a single leaf). 
	
	If the formula is the logical \textsc{And} of two subformulae, then $y$ must evaluate to 1 on both subformulae but $x$ must evaluate to 0 on at least one, so the player holding $x$ can report which. The parties continue with the protocol for that subformula, so the number of leaves in the communication protocol is (by induction) the sum of the number of leaves in the subformulae, which is just the number of leaves in the entire formula. A similar situation holds when the formula is the logical \textsc{Or} of two subformulae, but with the player holding $y$ speaking. 
	
	A communication protocol for $\kw(f)$ can be used to construct a formula for $f$ in an analogous fashion.
\end{proof}

A communication protocol for $\kw(f)$ partitions $f^{-1}(0) \times f^{-1}(1)$ into $C^P(\kw(f))$ combinatorial rectangles, where each rectangle is \emph{monochromatic} in terms of $\kw(f)$: that is, each rectangle is associated with some $i$ where $x_i \not= y_i$ for all $(x, y)$ in the rectangle. Let $C^D(\kw(f))$ be the minimum number of monochromatic combinatorial rectangles required to partition $f^{-1}(0) \times f^{-1}(1)$. Clearly, $C^D(\kw(f)) \leq C^P(\kw(f))$.

In order to prove Theorem \ref{thm:adv-l-lb}, we will exploit two properties of the spectral norm. The first is that the spectral norm (indeed, any matrix norm) is monotone with respect to submatrices: if $A$ is a submatrix of $B$, then $\norm{A} \leq \norm{B}$. The second is that the \emph{square} of the spectral norm is subadditive over rectangles. For a $|X| \times |Y|$ matrix $A$ and a combinatorial rectangle $R \subseteq X \times Y$, let $A_R$ be defined by:
$$ \bra{x}A_R\ket{y} = \begin{cases}\bra{x}A\ket{y} & (x, y) \in R\\0 & \textup{otherwise}\end{cases} $$
\begin{lemma}[\cite{DBLP:journals/cc/LaplanteLS06}] \label{lem:spectral-norm-squared-subadditive}
	If $A$ is an $|X| \times |Y|$ matrix and $\mathcal{R}$ partitions $X \times Y$ into combinatorial rectangles, then $\norm{A}^2 \leq \sum_{R \in \mathcal{R}} \norm{A_R}^2$.
\end{lemma}

\begin{proof}
	Note that $\norm{A} = \max_{u, v} |\bra{u} A \ket{v}| / \norm{\ket{u}} \norm{\ket{v}}$. In the following, let $\ket{u}$ and $\ket{v}$ be the unit vectors that achieve the maximum in this expression. 
	
	For any $R \in \mathcal{R}$ where $R = X_R \times Y_R$ for $X_R \subseteq X, Y_R \subseteq Y$, define $\ket{u_R}$ and $\ket{v_R}$ as follows:
	
	$$ \ip{u_R}{x} = \begin{cases} \ip{u}{x} & x \in X_R \\ 0 & \textup{otherwise} \end{cases} \qquad
	\ip{v_R}{y} = \begin{cases} \ip{v}{y} & y \in Y_R \\ 0 & \textup{otherwise} \end{cases} $$
	\begin{align*}
		\norm{A} &= |\bra{u} A \ket{v}| = \left| \bra{u} \left(\sum_{R \in \mathcal{R}} A_R \right)\ket{v} \right| = \left| \sum_{R \in \mathcal{R}} \bra{u} A_R \ket{v} \right| = \left| \sum_{R \in \mathcal{R}} \bra{u_R} A_R \ket{v_R} \right|\\
		& \leq \sum_{R \in \mathcal{R}} \left| \bra{u_R} A_R \ket{v_R} \right| \leq \sum_{R \in \mathcal{R}} \norm{A_R} |\ket{u_R}| |\ket{v_R}| \\
		& \leq \left( \sum_{R \in \mathcal{R}} \norm{A_R}^2\right)^{1/2} \left( \sum_{R \in \mathcal{R}} |\ket{u_R}|^2 |\ket{v_R}|^2 \right)^{1/2} \quad \mbox{(by the Cauchy-Schwarz inequality)}
	\end{align*}
	Note that the second term here simplifies:
	\begin{align*}
		\sum_{R \in \mathcal{R}} |\ket{u_R}|^2 |\ket{v_R}|^2 &= \sum_{R \in \mathcal{R}} \sum_{(x, y) \in R} (\ip{u}{x})^2 (\ip{v}{y})^2 \\
		&= |\ket{u}|^2 |\ket{v}|^2 \quad \mbox{(as $\mathcal{R}$ partitions $X \times Y$)}
	\end{align*}
	To conclude, note that as $\ket{u}$ and $\ket{v}$ are unit vectors, $|\ket{u}|^2 |\ket{v}|^2 = 1$: therefore, $\norm{A} \leq \left( \sum_{R \in \mathcal{R}} \norm{A_R}^2\right)^{1/2}$ and so $\norm{A}^2 \leq \sum_{R \in \mathcal{R}} \norm{A_R}^2$.
\end{proof}

Now we can prove Theorem \ref{thm:adv-l-lb}.

\begin{proof}[Proof of Theorem \ref{thm:adv-l-lb}]
	Let $A$ be any $f^{-1}(0) \times f^{-1}(1)$ matrix. Let $\mathcal{R}_f$ be an optimal rectangle partition in terms of $\kw(f)$.
	$$ \norm{A}^2 \leq \sum_{R \in \mathcal{R}_f} \norm{A_R}^2 \leq C^D(\kw(f)) \cdot \max_{R \in \mathcal{R}_f} \norm{A_R}^2 $$
	Let $A_i$ be the $f^{-1}(0) \times f^{-1}(1)$ matrix defined by:
	$$ \bra{x}A_i\ket{y} = \begin{cases}\bra{x}A\ket{y} & x_i \not= y_i\\0 & \textup{otherwise}\end{cases} $$
	Then, for any $R$, $A_R$ is a submatrix of $A_i$, so by the monotonicity with respect to rectangles:
	$$ C^D(\kw(f)) \cdot \max_{R \in \mathcal{R}_f} \norm{A_R}^2 \leq C^D(\kw(f)) \cdot \max_{i \in [n]} \norm{A_i}^2 $$
	Rearranging, we get:
	$$\mathcal{L}(f) \geq C^D(\kw(f)) \geq \max_{A \not= 0} \frac{\norm{A}^2}{\max_i \norm{A_i}^2}$$
	We conclude by taking the square root of the above expression and noting that for any matrix $A \in f^{-1}(0) \times f^{-1}(1)$, letting $A'$ be the matrix of the form $A' = \begin{bmatrix}0&A\\A^*&0\end{bmatrix}$, we have that $A'$ is an adversary matrix for $f$ and $\norm{A'} = \norm{A}$, so maximizing over matrices $A$ on the right-hand side is equivalent to maximizing over adversary matrices $A'$.
\end{proof}
\fi %

\section{Span Programs}
Given a function $f$, its \emph{span program} $P_f$ is an algebraic model of computation for $f$ first introduced by Karchmer and Wigderson. Let $A$ be a matrix whose columns are labelled by the set of $2n$ literals. Let $\ket{t}$ be a target vector. For input $s$ such that $f(s) = 1$, we would like the target vector $\ket{t}$ to be contained in the span of the columns of $A$ that are labelled by literals that agree with $s$. Otherwise, if $f(s) = 0$, we require that $\ket{t}$ is not in the span. In this case, there must exist some vector $\bra{y}$ that witnesses this fact (as shown by Farkas' Lemma). More formally, we define the span program as follows. 

\begin{definition}{\textup{(Span Programs, \cite{DBLP:conf/coco/KarchmerW93}.)}}
    A \emph{span program} $P_f$ for a $n$-ary boolean function $f$ consists of a matrix $A \in \ltrans(\CC^{I}, \CC^{[m]})$ and a target vector $\ket{t} \in \CC^{[m]}$, where $I$ is the disjoint union of $2n$ index sets $I_{1,0}$, $I_{1,1}$, ..., $I_{n,0}$, $I_{n,1}$ one for each setting of each entry of a boolean string $s \in \{0,1\}^n$. Given $s$, $\Pi(s) \in \ltrans(\{0,1\}^{I})$ is the diagonal matrix whose diagonal entry $(j,b)\times(j,b)$ indicates $s_j = b$ i.e.
    \begin{equation}
        \label{eq:input-matrix}
        \Pi(s) = \id - \sum_{j \in [n]} \ket{j, \overline{s}_j}\bra{j, \overline{s}_j}.
    \end{equation}
    The span program $P_f$ evaluates to false if there exists a negative witness $\ket{y} \in \CC^{[m]}$ i.e. $\bra{y}A\Pi(s) = 0$ but $\ip{y}{t} > 0$; wlog.\ assume that $\ip{y}{t} = 1$ by scaling. Conversely, $P_f$ evaluates to true if there exists a positive witness $\ket{z} \in \CC^{I}$ i.e. $\ket{t}$ is in the span of $A\Pi(s)$ and $A\Pi(s)\ket{z} = \ket{t}$. 
    
    For inputs which evaluate to true on $P_f$, i.e. $x \in F_1$ for which there exists $\ket{z}$ such that $A\Pi(x)\ket{z} = \ket{t}$, let $\wsize(P_f, x) = \norm{\ket{z}}^2$. For inputs which evaluate to false on $P_f$, i.e. $w \in F_0$ for which there exists $\ket{y}$ such that $\bra{y}A\Pi(w) = 0$ and $\ip{y}{t} = 1$, let $\wsize(P_f, w) = \norm{\bra{y}A}^2$.\footnote{Note that this value is equivalent to $\norm{\bra{z}A(\id - \Pi(w))}$ by the first condition.} The \emph{witness size} of $P_f$ is then
    \[\wsize(P_f) = \max_{s \in \{0,1\}^{n}} \wsize(P_f, s).\] 
\end{definition}

Using SDP duality, we show that the witness size of the span-program of $f$ is equivalent to its general adversary bound.

\ifdefined\fullversion
\begin{example}
Consider the span programs for several simple functions. Note that there can be many different span programs for the same function. All omitted column index sets are assumed to be empty.
\begin{enumerate}
        \item For the $n$-ary logical or function, $\op{OR}_n$, let $\ket{t} = [1]$ and  
        \[A = \begin{blockarray}{ccccc}
        I_{1,1} & I_{2,1} & \cdots & I_{n-1,1} & I_{n,1}\\
        \begin{block}{[ccccc]}
        1 & 1 & \cdots & 1 & 1\\
        \end{block}
        \end{blockarray}.
        \]
        Observe that $\wsize(P_{\op{OR}_n}) = \max_{s \in \{0,1\}^{n}} \wsize(P_{\op{OR}_n}, s) = n^2$ is achieved by the input string $s = [0,...0]^{\top}$ with the witness $\ket{y} = [1]$.
        \item For the parity function $\parity$, let $\ket{t} = [1,1]^{\top}$ and
        \[A = \begin{blockarray}{ccccc}
        I_{1,0} & I_{1,1} & I_{2,0} & I_{2,1}\\
        \begin{block}{[ccccc]}
        1 & 0 & 1 & 0\\
        0 & 1 & 0 & 1\\
        \end{block}
        \end{blockarray}.
        \]
        Observe that $\wsize(P_{\parity}) = 2$. This can be achieved by a string which evaluates to false e.g. $w = 00$ with witness $\ket{y} = [0,1]^{\top}$ and $\wsize(P_{\parity}, w) = \norm{\ket{y}A}^2$ or by a string which evaluates to true e.g. $x = 01$ with witness $\ket{z} = [1,1]^{\top}$ and $\wsize(P_{\parity}, w) = \norm{\ket{z}}^2$.
    \end{enumerate}
\end{example}
\else
    See Appendix~\ref{sec:span-program-example} for an example span programs for the parity function. Note that there can be many different span programs for the same function. 
\fi

\subsection{Canonical Span Programs}
In order to relate the complexity of the span program of a given function $f$ to its query complexity, we put it in \emph{canonical span program} form. Every span program can be transformed into a canonical span program with at most a polynomial blow-up in size \cite{DBLP:conf/coco/KarchmerW93}.

\begin{definition}{\textup{(Canonical Span Program.)}} 
    The input matrix $A$ and target vector $\ket{t}$ of the canonical span program will be as follows. Define $\ket{t} \in \CC^{F_0}$ to be a scalar multiple of the all ones vector. Let $A \in \mathcal{L}(\CC^{I}, \CC^{F_0})$ where $I = [n] \times \{0,1\} \times [m]$ for a yet-to-be-determined $m$. Each row of $A$ corresponds to an input $w$ evaluating to zero on $f$. Divide this row further into $2n$ row vectors of length $m$ one for each setting of each entry in the input. In the following, if $w_j = b$, then denote each length-$m$ vector corresponding to $I_{j,b}(w)$ by $\ket{v_{w,j}'}$ and corresponding to $I_{j, \overline{b}}$ by $\ket{v_{w,j}}$. 

    Define $\ket{v_{w,j}'}$ to be the all zeroes vector for all $w \in F_0$ and $j \in [m]$. Observe that $\ket{t}$ cannot be in the span of $A\Pi(w)$ since the row of $A\Pi(w)$ corresponding to $w$ consists entirely of zeros. Further, since the indicator vector $\ket{w} \in \CC^{F_0}$ for $w$ is a witness for $A$,\footnote{Since $\ip{w}{t} = 1$ while $\bra{w}A\Pi(w) = 0$.}
    \[\wsize(P_f, w) = \norm{\bra{w}A}^2 = \sum_{j \in [n]} \norm{\ket{v_{w,j}}}^2.\]

    Each $x \in F_1$ will be assign an input vector of length $mn$. These will \emph{not} appear in $A$, but will be used to ensure that the vectors $\ket{v_{w,j}}$ in $A$ satisfy certain constraints. Each vector will be divided into $n$ length $m$ vectors corresponding to the $n$ entries of $x$. These will be denoted by $\ket{v_{x, j}}$. Since $\ket{t}$ needs to be in the span of $A\Pi(x)$, we require that for all $w \in F_0$, $\sum_{w_j \neq x_j} \ip{v_{w, j}}{v_{x, j}} = 1$. Observe that the witness size is again of the form 
    \[\wsize(P_f, x) = \sum_{j \in [n]} \norm{\ket{v_{x,j}}}^2.\]
    
    The smallest $m$ for which there exists such vectors $\ket{v_{w,j}}$ and $\ket{v_{x,j}}$ will suffice. 
\end{definition}
In the following let $W$ be the witness size of the canonical span program.
\ifdefined\fullversion
\begin{example}
    \label{ex:span-program-for-parity}
    The canonical span program for $\parity$ is as follows. Let the target vector be $\ket{t} = c[1,1]^{\top}$ where $c = 1/(3\sqrt{W})$. Then for $\{w_1=00, w_2=11\} = F_0$ with vector $\bra{v_{w_i,j}} \in \CC^{[m]}$ corresponding to the length $m$ vector of the $j$\textsuperscript{th} bit of $w_i$, we have
    \[A = \begin{blockarray}{ccccc}
        I_{1,0} & I_{1,1} & I_{2,0} & I_{2,1}\\
        \begin{block}{[cccc]c}
        0 & \bra{v_{w_1,1}} & 0 & \bra{v_{w_1,2}} & w_1 = 00\\
        \bra{v_{w_2,1}} & 0 & \bra{v_{w_2,2}} & 0 & w_2 = 11\\
        \end{block}
        \end{blockarray}
    \]
    Further, to each string $x_i$ in $\{x_1 = 10, x_2 = 01\} = F_1$ we assign a vector $\ket{x_i} = [\ket{v_{x_i,1}}, \ket{v_{x_i,2}}]^{\top}$ where $\ket{v_{x_i, j}} \in \CC^{[m]}$ corresponds to the length $m$ vector of the $j$\textsuperscript{th} bit of $x_i$. 
    Note that $m = 1$ suffices, since the matrix $A$ where
        \[A = \begin{blockarray}{ccccc}
        I_{1,0} & I_{1,1} & I_{2,0} & I_{2,1}\\
        \begin{block}{[cccc]c}
        0 & 1 & 0 & 1 & w^{(1)} = 00\\
        1 & 0 & 1 & 0 & w^{(2)} = 11\\
        \end{block}
        \end{blockarray}
    \]
    and the pair of vectors $\ket{x_1} = \ket{x_2} = [1, 1]^{\top}$ satisfies the condition $\sum_{w_j \neq x_j} \ip{v_{w,j}}{v_{x,j}} = 1$.
\end{example}
\fi

\subsection{The Dual of \texorpdfstring{$\adv$}{Adv+-} is Span Program Witness Size}
From the canonical span program above we write the witness size as the following optimization problem:
\[\wsize(P_f) = \min_{\{\ket{v_{x,j}}\}}\max_{s \in \{0,1\}^n, j \in [n]} \norm{\ket{v_{s,j}}}^2\]
subject to the constraint that for all pairs $(w,x) \in \advmatrix$, $\sum_{w_j \neq x_j} \ip{v_{w,j}}{v_{x,j}} = 1$.
Let $X$ be PSD matrix such that entry $\bra{w,i}X\ket{x,j} = \ip{v_{w,i}}{v_{x,j}}$ for all $w \in F_0$ and $x \in F_1$. Write $\wsize(P_f)$ as the following equivalent SDP
\[\wsize(P_f) = \min_{X \succeq 0} \max_{s \in \{0,1\}^{n}} \sum_{j \in [n]} \bra{s,j}X\ket{s,j}\]
subject to the constraint that for all $(x,w) \in \advmatrix$, $\sum_{w_j \neq x_j}\bra{w,j}X\ket{x,j} = 1$.

We will turn the above SDP into the general adversarial bound. First introduce a variable $\xi$ in order to eliminate the inner maximization function. For adversary matrix $\Gamma$ let $\Gamma_j = \Gamma \circ D_j$.
\begin{align}
    \wsize(P_f) &= \min_{\substack{
    X \succeq 0\mbox{, } \xi \geq 0, \\
    \forall (w,x) \in \advmatrix: \sum_{w_j \neq x_j}\bra{w,j}X\ket{x,j} = 1,\\
    \forall s \in \{0,1\}^n: \xi \geq \sum_{j \in [n]} \bra{s,j}X\ket{s,j}}} \xi \label{eq:primal-sdp}\\
    &= \max_{\substack{
    \{\alpha_{w,x}\}, \\
    \beta_s \geq 0\mbox{, }\sum_{s} \beta_s = 1,\\
    \sum_{s}\beta_s\ket{s}\bra{s} \succeq \sum_{w,x \in \advmatrix, w_j \neq x_j}\alpha_{w,x}\ket{w}\bra{x}}} \sum \alpha_{w,x} &\left(\mbox{SDP duality; see Appendix~\ref{sec:lagrangian-duality}}\right) \label{eq:dual-sdp}\\
    &= \max_{\substack{
    \{\alpha_{w,x}\},\\
    \beta_s \geq 0\mbox{, }\sum_{s} \beta_s = 1,\\
    \sum_{s \in \{0,1\}^n}\ket{s'}\bra{s'} \succeq \sum_{w,x \in \advmatrix, w_j \neq x_j}\frac{\alpha_{w,x}}{\sqrt{\beta_w\beta_x}}\ket{w'}\bra{x'}}} \sum \alpha_{w,x} &\left(\mbox{substitute }\ket{s'} = \frac{1}{\sqrt{\beta_s}}\ket{s}\right)\\
    &= \max_{\substack{
    \{\alpha_{w,x}'\},\\
    \beta_s \geq 0\mbox{, }\sum_{s} \beta_s = 1,\\
    \sum_{s \in \{0,1\}^n}\ket{s'}\bra{s'} \succeq \sum_{w,x \in \advmatrix, w_j \neq x_j}\alpha_{w,x}'\ket{w'}\bra{x'}}} \sum \alpha_{w,x}'\sqrt{\beta_w\beta_x} &\left(\mbox{substitute }\alpha_{w,x}' = \alpha_{w,x}/\sqrt{\beta_w\beta_x}\right)\\
    &= \max_{\substack{
    \Gamma_{w,x} = \alpha_{w,x}',\\
    \ket{\beta}_s = \sqrt{\beta_s}\mbox{, }\norm{\ket{\beta}} = 1,\\
    \id - \Gamma_i \succeq 0}} \bra{\beta}\Gamma\ket{\beta}&\left(\mbox{$\bra{w}\Gamma_j\ket{x} = 0$ if $w_j = x_j$} \right)\\
    &= \max_{\substack{
    \Gamma_{w,x} = \alpha_{w,x}',\\
    \norm{\Gamma_i} \leq 1}} \norm{\Gamma} = \adv(f)
\end{align}

\subsection{Span Programs as Graphs}
The canonical span program matrix $A$ of $f$ can be transformed into the biadjacency matrix of two bipartite graphs \cite{DBLP:journals/eccc/Reichardt10a, DBLP:conf/soda/Reichardt11a, DBLP:journals/toc/ReichardtS12}. These graphs capture the evaluation of a string $s$ on $f$: in particular, we define a ``true'' biadjacency matrix such that if $f(s) = 1$ then there is an eigenvalue-zero eigenvector while no such eigenvector exists when $f(s) = 0$, and a ``false'' biadjacency matrix where the opposite is true. Let $B_{G(s)} \in \CC^{(F_0 \cup I') \times (\{\mu_0\} \cup I)}$ and $B_{G'(s)} \in \CC^{(F_0\cup I') \times I}$ be the true and false biadjacency matrices corresponding to the bipartite graph $G$ of the span program respectively:
\begin{equation}
    \label{eq:biadjacency-matrices}
    B_{G(s)} = 
    \begin{blockarray}{ccc}
        \mu_0 & I &\\
        \begin{block}{[cc]c}
        \ket{t} & A & F_0\\
        0 & \overline{\Pi}(s) & I'\\
        \end{block}
    \end{blockarray}\qquad
    B_{G'(s)} = 
        \begin{blockarray}{cc}
            I &\\
            \begin{block}{[c]c}
            A & F_0\\
            \overline{\Pi}(s) & I'\\
            \end{block}
        \end{blockarray}
\end{equation}
where $\ket{t}$ and $A$ are defined as 
\begin{equation}
    \label{eq:inputs-to-canonical-span-program}
    \ket{t} = \frac{1}{3\sqrt{W}}\sum_{w \in F_0} \ket{w} \mbox{ and } A = \sum_{w \in F_0, j \in [n]}\ketbra{w}{j,\overline{w}_j} \tensor \ket{v_{w,j}}
\end{equation}
where $W$ is the witness size and $\overline{\Pi}(s) = \id - \Pi(s) \in \ltrans(\CC^{I})$ (see $\Pi(s)$ in Equation~\ref{eq:input-matrix}).

Matrix-vector products $B_{G(s)}\ket{\psi}$ and $B_{G'(s)}^{*}\ket{\psi'}$ can be interpreted as operating on the sets of column vectors separately. That is, let $\ket{\psi} = \ket{\psi_1} + \ket{\psi_2}$ where $\ket{\psi_1} := \alpha\ket{0}$ operates on column $\mu_0$ and $\ket{\psi_2}$ operates on columns $I$. Similarly, let $\ket{\psi'} = \ket{\psi_1'} + \ket{\psi_2'}$ where $\ket{\psi_1'}$ and $\ket{\psi_2'}$ operates on rows $F_0$ and $I'$.  

\ifdefined\fullversion
\begin{example}
    Let us turn the canonical span program of the parity function, shown in Example \ref{ex:span-program-for-parity}, into its corresponding bipartite graphs. The matrices $B_{G(x)}$ and $B_{G'(w)}$ are then defined as follows for strings $x = 10$ and $w = 00$ which evaluates to true and false respectively.
    \[B_{G(x)} = \begin{blockarray}{cccccc}
    \mu_0 & \BAmulticolumn{4}{c}{I}\\
    \begin{block}{[c|cccc]c}
    1 & 0 & 1 & 0 & 1 & \multirow{2}{*}{$F_0$}\\
    1 & 1 & 0 & 1 & 0 & \\\cline{1-5}
    0 & 1 & 0 & 0 & 0 & \multirow{4}{*}{$I'$}\\
    0 & 0 & 0 & 0 & 0 &\\
    0 & 0 & 0 & 0 & 0 &\\
    0 & 0 & 0 & 0 & 1 &\\
    \end{block}
    \end{blockarray}\qquad
    B_{G'(w)}^{*} = 
    \begin{blockarray}{ccccccc}
    \BAmulticolumn{2}{c}{F_{0}} & \BAmulticolumn{4}{c}{I'}\\
    \begin{block}{[cc|cccc]c}
    0 & 1 & 0 & 0 & 0 & 0 & \multirow{4}{*}{I}\\
    1 & 0 & 0 & 1 & 0 & 0 & \\
    0 & 1 & 0 & 0 & 0 & 0 & \\
    1 & 0 & 0 & 0 & 0 & 1 & \\
    \end{block}
    \end{blockarray}.\]
    These corresponds to the bipartite graphs shown Figure~\ref{fig:biadjacencymatrix-true} and Figure~\ref{fig:biadjacencymatrix-false}. 
    \fig{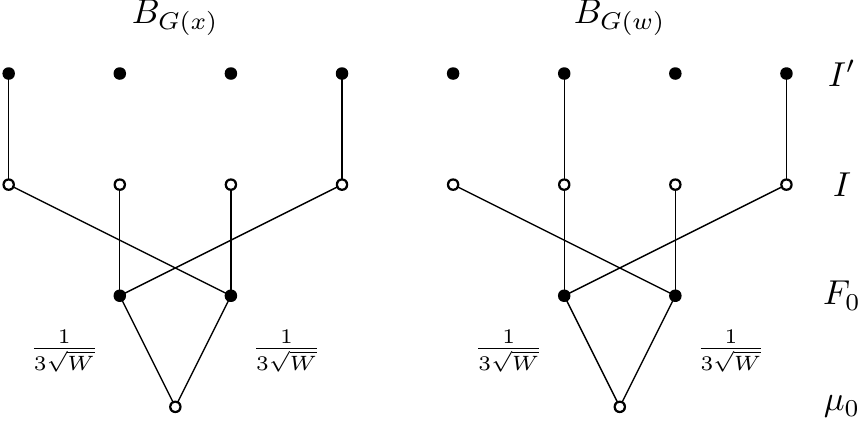}{1}{The bipartite graphs corresponding to the true biadjacency matrix. All unmarked edges have weight one. The matix vector product $B_{G(x)} \ket{\phi}$ is equivalent to assigning weights to the open dots in the picture. In order to find an eigenvalue zero eigenvector $\ket{\phi}$ of $B_{G(x)}$, the assignment of weights must ensure the neighbours of every solid dot sums to zero. Observe that $B_{G(x)}$, with $\oplus(x) = 1$, has an eigenvalue zero eigenvector while $B_{G(w)}$, with $\oplus(w) = 0$, does not. }
    \fig{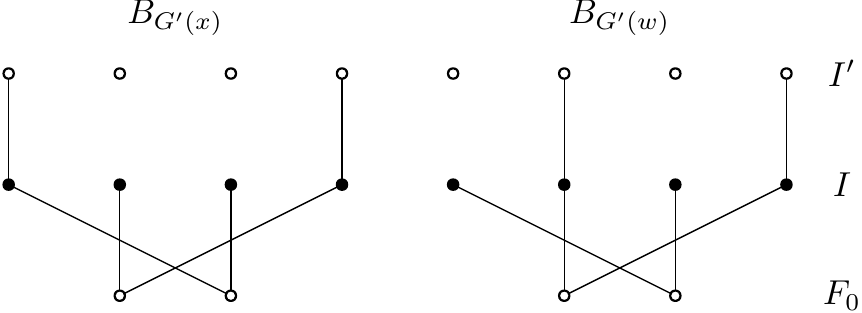}{1}{The bipartite graphs corresponding to the false biadjacency matrix. As opposed to the above, $B_{G'(x)}$, with $\oplus(x) = 1$, does not have an eigenvalue zero eigenvector while $B_{G'(w)}$, with $\oplus(w) = 0$, does.}
\end{example}
\fi

\begin{lemma}{\textup{(Spectral Gap of Eigenvalue Zero Eigenvectors.)}}
\label{lem:spectra-gap}
If $f(x) = 1$, then the vector 
\[\ket{\psi} = \ket{\psi_1} + \ket{\psi_2} \mbox{ where }\ket{\psi_1} = -3\sqrt{W}\ket{0}\mbox{ and }\ket{\psi_2} = \sum_{j \in [n]}\ket{j, x_j} \tensor \ket{v_{x, j}}\] 
is an eigenvalue zero eigenvector of $B_{G(x)}$. Further, $\ket{\psi}$ satisfies $|\ip{0}{\psi}|^2 \geq 9\norm{\ket{\psi}}^2/10$.

If instead $f(w) = 0$, then the vector 
\[\ket{\psi'} = \ket{\psi_1'} + \ket{\psi_2'}\mbox{ where }\ket{\psi_1'} = -\ket{w}\mbox{ and }\ket{\psi_2'} = \sum_{j \in [n]}\ket{j, \overline{w}_j} \tensor \ket{v_{w, j}}\] 
is an eigenvalue zero eigenvector of $B_{G'(w)}$. Further, $\ket{\psi'}$ satisfies $|\ip{t}{\psi'}|^2 \geq \norm{\ket{\psi}}^2/(9W(W + 1))$. 
\end{lemma}
\begin{proof}
    Let $x \in F_1$ and $w \in F_0$. Observe that $B_{G(x)}\ket{\psi_1}$ is the vector with $|F_0|$ non-zero entries followed by $|I'|$ zeros. Since $f(x) = 1$, there exists a linear combination of the columns of $A$ which sum to $\ket{t}$. Choosing this set of columns will also ensure that the rows indexed by $I'$ sum to zero as an entry of $\ket{\psi_2}$ is non-zero only when the associated column of $\overline{\Pi}(x)$ is. Further $|\ip{0}{\psi}|^2 = \norm{\ket{\psi_1}}^2 = 9W$, while $\norm{\ket{\psi}}^2 = \norm{\ket{\psi_1}}^2 + \norm{\ket{\psi_2}}^2 = 9W + W$ by definition. Similarly, observe that $B_{G'(w)}^{*}\ket{\psi'_1}$ multiplies the column associated with $w$ among the rows of $A$ by negative one, while $B_{G'(w)}^{*}\ket{\psi'_2}$ is exactly this same column. Further, $|\ip{t}{\psi'}|^2 = 1/(9W)$ and $\norm{\ket{\psi'}}^2 = \norm{\ket{\psi_1'}}^2 + \norm{\ket{\psi_2'}}^2 = 1 + W$.
\end{proof}
\section{Optimal Quantum Query Algorithms for Span Programs}
Let $\ket{t}$ and $A$, as shown in Equation~\ref{eq:inputs-to-canonical-span-program}, be the target and matrix of the canonical span program respectively. Further let $G$ be the associated bipartite graph with biadjacency matrix $B_G$ and adjacency matrix $A_G$ as follows
\begin{equation}
    \label{eq:biadj-adj-of-G}
    B_G = 
    \begin{blockarray}{ccc}
    \mu_0 & I \\
    \begin{block}{[cc]c}
    \ket{t} & A & F_{0}\\
    \end{block}
    \end{blockarray}\mbox{ and }A_G = 
    \begin{blockarray}{cccc}
    F_0 & \mu_0 & I \\
    \begin{block}{[ccc]c}
    0 & \ket{t} & A & F_{0}\\
    \bra{t} & 0 & 0 & \mu_0\\
    A^{*} & 0 & 0 & I\\
    \end{block}
    \end{blockarray}
\end{equation}
Let $\Delta \in \ltrans(\CC^{F_0 \cup\{\mu_0\}\cup I})$ be the orthogonal projection onto the span of all eigenvalue zero eigenvectors of $A_G$. For a string $s \in \{0,1\}^n$, let $\Pi_s \in \ltrans(\CC^{F_0\cup\{\mu_0\}\cup I})$ be
\[\Pi_s = \id - \sum_{j \in [n], k \in [m]} \ket{j, \overline{s}_j, k}\bra{j, \overline{s}_j, k}.\]
The graph $G(s)$ has biadjacency matrix $B_{G(s)}$ (from Equation~\ref{eq:biadjacency-matrices}) and adjacency matrix $A_{G(s)}$.
\begin{equation}
    \label{eq:biadj-adj-of-Gs}
    B_{G(s)} = 
    \begin{blockarray}{ccc}
    \mu_0 & I \\
    \begin{block}{[cc]c}
    \ket{t} & A & F_{0}\\
    0 & \overline{\Pi}(s) & I'\\
    \end{block}
    \end{blockarray}\mbox{ and }A_{G(s)} = 
    \begin{blockarray}{ccccc}
    F_0 & I' & \mu_0 & I \\
    \begin{block}{[cccc]c}
    0 & 0 & \ket{t} & A & F_{0}\\
    0 & 0 & 0 & \overline{\Pi}(s) & I'\\
    \bra{t} & 0 & 0 & 0 & \mu_0\\
    A^{*} & \overline{\Pi}(s) & 0 & 0 & I\\
    \end{block}
    \end{blockarray}
\end{equation}
Note that $A_{G(s)} \in \ltrans(\CC^{F_{0}\cup I'\cup \{\mu_0\}\cup I})$ contains $A_{G}$ and the additional vertices of $I'$. Further $\id - \Pi_s \in \ltrans(\CC^{F_{0}\cup\{\mu_0\}\cup I})$ contains $\overline{\Pi}(s) \in \ltrans(\CC^{I})$ as a subgraph and is everywhere else zeros.

Define $U_s \in \ltrans(\CC^{F_0\cup\{\mu_0\}\cup I})$ as
\[U_s = (2\Pi_s - \id)(2\Delta - \id),\]
the matrix which reflects a vector across $\Delta$ then across $\Pi_s$. Observe that $\Delta$ is independent of the input $s$, while $\Pi_s$ requires one query of the quantum $f$-oracle. The following are three different quantum query algorithms which compute $f(s)$ with query complexity $W$. 

\begin{algorithm}[H]
    \caption{Phase Estimation}
    \label{alg:phase-estimation}
    \DontPrintSemicolon
    Initialize state $\ket{0} \in \CC^{F_0 \cup \mu_0 \cup I}$\;
    $\delta_{p} \leftarrow \frac{1}{100W}$\;
    $\delta_{e} \leftarrow \frac{1}{10}$\;
    Run phase estimation on $U_s$ with precision $\delta_{p}$ and error $\delta_{e}$\;
    Return $1$ if phase estimation returns zero, otherwise return $0$\;
\end{algorithm}

\begin{algorithm}[H]
    \caption{Quantum Search}
    \label{alg:quantum-search}
    \DontPrintSemicolon
    Initialize state $\ket{+}\tensor\ket{0} \in \CC^2 \tensor \CC^{F_0 \cup \mu_0 \cup I}$\;
    $T \leftarrow $random integer in $\{1, ..., \ceil{100W}\}$\;
    Apply $\ketbra{0}{0}\tensor\id + \ketbra{1}{1}\tensor U^{T}_{x}$ to initial state\;
    Measure the first qubit in the Hadamard basis\;
    Return $1$ if the value is $\ket{+}$, otherwise return $0$\;
\end{algorithm}

\begin{algorithm}[H]
    \caption{Quantum Search without Register}
    \label{alg:quantum-search-wo-register}
    \DontPrintSemicolon
    Initialize state $\ket{0} \in \CC^{F_0 \cup \mu_0 \cup I}$\;
    $T \leftarrow $random integer in $\{1, ..., \ceil{100W}\}$\;
    Apply $U^{T}_{x}$ to $\ket{0}$\;
    Measure $U^{T}_{x}\ket{0}$ in the standard basis\;
    Return $1$ if the value is $\ket{0}$, otherwise return $0$\;
\end{algorithm}

We will only analyse the first two algorithms. The analysis for the third  Algorithm~\ref{alg:quantum-search-wo-register} is quite complex and the quantum query complexity is equivalent to the other two. The following lemma about the ``effective spectral'' gap of $A_{G(s)}$ will be necessary for the analysis. Its intuition and proof can be found in Appendix~\ref{sec:small-eval-evectors}. 
\begin{lemma}{\textup{(Effective Spectral Gap.)}}
\label{lem:effective-spectral-gap}
If $f(s) = 1$ then $A_{G(s)}$ has an eigenvalue zero eigenvector $\ket{\psi}$ with $|\ip{0}{\psi}|^2 \geq 9\norm{\ket{\psi}}^2/10$. 

If $f(w) = 0$ and $\{\ket{\alpha}\}$ is the set of all orthonormal eigenvectors with corresponding eigenvalues $\rho(\alpha)$ of $A_{G(s)}$, then for any $c \geq 0$
\[\sum_{\alpha: |\rho(\alpha)| \leq c/W}|\ip{\alpha}{0}|^2 \leq 72c^2\left(1 + \frac{1}{W}\right).\]
\end{lemma}

\subsection{Spectral Gap for \texorpdfstring{$U_s$}{U(s)}}
Using Lemma~\ref{lem:effective-spectral-gap}, we prove a spectral gap on the eigenvectors of the matrix $U_s$.
\begin{lemma}
	\label{lem:effective-spectral-gap-reflections}
	If $f(s) = 1$ then $U_s$ has an eigenvalue one eigenvector $\ket{\varphi}$ with $|\ip{0}{\varphi}|^2 / \norm{\ket{\varphi}}^2 \geq 9/10$. 
	
	If $f(s) = 0$ and $\{\ket{\beta}\}$ is a set of orthonormal eigenvectors of $U_s$ with corresponding eigenvalues $e^{i \theta(\beta)}$, where $\theta(\beta) \in (-\pi, \pi]$. Then for any $\Theta \geq 0$
	\[ \sum_{\beta: |\theta(\beta)| \leq \Theta} |\ip{\beta}{0}|^2 \leq \left(2\sqrt{6\Theta W} + \frac{\Theta}{2}\right)^2 \]
\end{lemma}

A key tool used to prove Lemma~\ref{lem:effective-spectral-gap-reflections} is the fact that we can rotate the basis of $U_s$ so that it becomes block-diagonal with blocks of maximum dimension two\footnote{We can do this for any unitary made up of two reflections}. This was proved by Szegedy~\cite{DBLP:conf/focs/Szegedy04}. Nagaj, Wocjan, and Zhang~\cite{DBLP:journals/qic/NagajWZ11} gave a different proof that follows from a Lemma of Jordan~\cite{jordan1875essai}:

\begin{lemma}{\textup{(\cite{jordan1875essai})}}
	\label{lem:jordan}
	Given projections $\Pi_s$ and $\Delta$ in Hilbert space $\mathcal{H}$, there exists a decomposition of $\mathcal{H}$ into orthogonal one-dimensional and two-dimensional subspaces invariant under $\Pi_s$ and $\Delta$. On the two-dimensional subspaces, $\Pi_s$ and $\Delta$ are rank-one projectors.
\end{lemma}

Lemma~\ref{lem:jordan} implies that $\mathcal{H}$ can be decomposed into a set of one-dimensional subspaces $\{T_i\}$ and a set of two-dimensional subspaces $\{S_i\}$. Each one-dimensional subspace $T_i$ is spanned by a vector $\ket{v_i}$ for which there exists $b, c \in \{0, 1\}$ such that $\Delta \ket{v_i} = b \ket{v_i}$ and $\Pi_s \ket{v_i} = c \ket{v_i}$: that is, each of $\Delta$ and $\Pi_s$ either act as the identity on $T_i$ or are orthogonal to $T_i$. Each two-dimensional subspace $S_i$ is spanned by vectors $\ket{v_i}, \ket{v_i^\bot}$ such that $\Delta \ket{v_i} = \ket{v_i}$ and $\Delta \ket{v_i^\bot} = 0$. Also, $S_i$ is spanned by vectors $\ket{w_i}, \ket{w_i^\bot}$ such that $\Pi_s \ket{w_i} = \ket{w_i}$ and $\Pi_s \ket{w_i^\bot} = 0$. Let $\theta_i = 2 \arccos |\ip{v_i}{w_i}|$. Then,
\begin{align*}
	\ket{w_i} = \cos \frac{\theta_i}{2} \ket{v_i} + \sin \frac{\theta_i}{2} \ket{v_i^\bot} \qquad \qquad
	\ket{w_i^\bot} = -\sin \frac{\theta_i}{2} \ket{v_i} + \cos \frac{\theta_i}{2} \ket{v_i^\bot}
\end{align*}
\begin{theorem}{\textup{(\cite{DBLP:conf/focs/Szegedy04, DBLP:journals/qic/NagajWZ11})}}
	\label{thm:decomposition}
	Let $\{S_i\}, \{T_i\}$ be the decomposition of $\Pi_s$ and $\Delta$ given by Lemma~\ref{lem:jordan}. Then $U_s$ has eigenvalues $e^{\mp i \theta_i}$ corresponding to $\frac{\ket{v_i} \pm i \ket{v_i^\bot}}{\sqrt{2}}$ on each two-dimensional subspace $S_i$, and eigenvalue either 1 or -1 on each one-dimensional subspace $T_i$.
\end{theorem}

\ifdefined\fullversion
\begin{proof}

	On one-dimensional subspace $T_i$, each individual reflection multiplies a vector by $\pm 1$, so both reflections in succession do as well. For the rest of the proof, consider a two-dimensional subspace $S_i$. By the above relationship between $\{\ket{v_i}, \ket{v_i^\bot}\}$ and $\{\ket{w_i}, \ket{w_i^\bot}\}$, we get the following:
	\begin{align*}
		\begin{bmatrix}\ket{w_i}\\\ket{w_i^\bot}\end{bmatrix} = \begin{bmatrix}\cos \frac{\theta_i}{2}&\sin \frac{\theta_i}{2}\\-\sin \frac{\theta_i}{2}&\cos \frac{\theta_i}{2}\end{bmatrix} \begin{bmatrix}\ket{v_i}\\\ket{v_i^\bot}\end{bmatrix}
	\end{align*}
	Recall that $\sigma_y = \begin{bsmallmatrix}0&-i\\i&0\end{bsmallmatrix}$ is the Pauli Y matrix, which has eigenvalues 1 and $-1$ corresponding to eigenvectors $\ket{\phi_y^+} = \frac{1}{\sqrt{2}} \begin{bsmallmatrix}1\\i\end{bsmallmatrix}$ and $\ket{\phi_y^-} = \frac{1}{\sqrt{2}} \begin{bsmallmatrix}1\\-i\end{bsmallmatrix}$, respectively.
	\begin{align*}
		\begin{bmatrix}\cos \frac{\theta_i}{2}&\sin \frac{\theta_i}{2}\\-\sin \frac{\theta_i}{2}&\cos \frac{\theta_i}{2}\end{bmatrix} &= \cos \frac{\theta_i}{2} \begin{bmatrix}1&0\\0&1\end{bmatrix} + i \sin \frac{\theta_i}{2} \begin{bmatrix}0&-i\\i&0\end{bmatrix}\\
		&=\frac{\cos \frac{\theta_i}{2}}{2}\left(\begin{bmatrix}1&-i\\i&1\end{bmatrix} + \begin{bmatrix}1&i\\-i&1\end{bmatrix} \right) + \frac{i \sin \frac{\theta_i}{2}}{2}\left(\begin{bmatrix}1&-i\\i&1\end{bmatrix} - \begin{bmatrix}1&i\\-i&1\end{bmatrix} \right)\\
		&=\left( \frac{\cos \frac{\theta_i}{2}}{2} + \frac{i \sin \frac{\theta_i}{2}}{2} \right) \begin{bmatrix}1&-i\\i&1\end{bmatrix} + \left( \frac{\cos \frac{\theta_i}{2}}{2} - \frac{i \sin \frac{\theta_i}{2}}{2} \right) \begin{bmatrix}1&i\\-i&1\end{bmatrix}\\
		&=e^{i \theta_i/2} \frac{1}{2} \begin{bmatrix}1&-i\\i&1\end{bmatrix} + e^{-i \theta_i/2} \frac{1}{2} \begin{bmatrix}1&i\\-i&1\end{bmatrix}\\
		&=e^{i \theta_i/2} \ket{\phi_y^+}\bra{\phi_y^+} + e^{-i \theta_i/2} \ket{\phi_y^-}\bra{\phi_y^-} = e^{(i \theta/2) \sigma_y}
	\end{align*}
	In the basis $\{\ket{v_i}, \ket{v_i^\bot}\}$, we have that $(2 \Delta - \id) = \begin{bsmallmatrix}2&0\\0&0\end{bsmallmatrix} - \begin{bsmallmatrix}1&0\\0&1\end{bsmallmatrix}=\begin{bsmallmatrix}1&0\\0&-1\end{bsmallmatrix}=\sigma_z$, where $\sigma_z$ is the Pauli Z matrix. Similarly, in the basis $\{\ket{w_i}, \ket{w_i^\bot}\}$, we have that $(2 \Pi(s) - \id) = \sigma_z$. Then, in the basis $\{\ket{v_i}, \ket{v_i^\bot}\}$, $U_s = (2 \Pi(s) - \id) (2 \Delta - \id) = e^{(-i \theta_i/2) \sigma_y} \sigma_z e^{(i \theta_i/2) \sigma_y} \sigma_z = e^{(-i \theta_i/2) \sigma_y} e^{(-i \theta_i/2) \sigma_y} \sigma_z \sigma_z = e^{(-i \theta_i) \sigma_y}$, where we used the fact that $\sigma_y$ and $\sigma_z$ anticommute ($\sigma_y \sigma_z = -\sigma_y \sigma_z$) and the fact that $\sigma_z \sigma_z = \id$. Therefore, the eigenvalues are $e^{-i \theta_i}$ and $e^{i \theta_i}$, corresponding to eigenvectors $\ket{\phi_y^+} = \frac{\ket{v_i} + i \ket{v_i^\bot}}{\sqrt{2}}$ and $\ket{\phi_y^-} = \frac{\ket{v_i} - i \ket{v_i^\bot}}{\sqrt{2}}$ respectively.
\end{proof}
\else %
The proof is in Appendix~\ref{sec:decomp-proof}.
\fi Now we can prove Lemma~\ref{lem:effective-spectral-gap-reflections}.

\begin{proof}[Proof of Lemma~\ref{lem:effective-spectral-gap-reflections}.]
	Let $\{\ket{\beta}\}$ be the set of eigenvectors given by the decomposition of Theorem~\ref{thm:decomposition}. Since $\Delta$ is the projection into the nullspace of $A_G$, $A_G\Delta = 0$. Thus $A_{G(s)}(\Delta \oplus \id) = T(\id - \Pi_s)\tensor(\id_2)$ for a permutation matrix $T$ since $G$ is a subgraph of $A_{G(s)}$ and $\overline{\Pi}(s)$ is a submatrix of $\id - \Pi_s$ from Equation~\ref{eq:biadj-adj-of-G} and Equation~\ref{eq:biadj-adj-of-Gs}. 
	
	First consider the case where $f(s) = 1$. Take $\ket{\psi}$ to be the eigenvalue zero eigenvector of $A_{G(s)}$ such that $|\ip{0}{\psi}|^2 \geq 9\norm{\ket{\psi}}^2/10$ from Lemma~\ref{lem:effective-spectral-gap}. Obtain $\ket{\phi}$ from $\ket{\psi}$ by restricting to the entries corresponding to the index sets $F_0 \cup \{\mu_0\} \cup I$. Since $\ket{\psi}$ is an eigenvalue zero eigenvector of $A_{G(s)}$ (see Lemma \ref{lem:spectra-gap}), it is not supported on the removed entries so $\norm{\ket{\psi}} = \norm{\ket{\phi}}$ and $\ket{\phi}$ is an eigenvalue zero eigenvector of $A_G$. Thus $\Delta\ket{\phi} = \ket{\phi}$. Since $\Pi_s$ is the identity matrix on the support of $\ket{\psi}$, $\Pi_s\ket{\phi} = \ket{\phi}$. Together $U_s\ket{\phi} = \ket{\phi}$. 
	
	Now consider the case where $f(s) = 0$. Let $\ket{\zeta} = \sum_{\beta: |\theta(\beta)| \leq \Theta} \ket{\beta}\ip{\beta}{0}$: this is the projection of $\ket{0}$ onto low-angle subspaces of $U_s$. We want to bound $\sum_{\beta: |\theta(\beta)| \leq \Theta} |\ip{\beta}{0}|^2 = \sum_{\beta: |\theta(\beta)| \leq \Theta} \ip{0}{\beta} \ip{\beta}{0} = \ip{0}{\zeta}$. We will find it more convenient to bound $|\ip{0}{\hat{\zeta}}|^2 = \ip{0}{\zeta}$, where $\ket{\hat{\zeta}}$ is the normalized vector $\ket{\zeta}/\norm{\ket{\zeta}}$.

	Observe that $\ket{\hat{\zeta}}$ is not supported on any eigenvectors $\ket{\beta}$ where $\theta(\beta) = 0$. Without loss of generality, $\theta(\beta) = 0$ only when $\ket{\beta}$ is in a one-dimensional subspace $T_i$ with eigenvalue one. Then $(2\Pi_s - \id)$ and $(2 \Delta - \id)$ either both reflect $\ket{\beta}$ or they both don't. In the first case, $\Pi_s \ket{\beta} = \Delta \ket{\beta} = 0$, so $\ip{0}{\beta} = \bra{0}\Pi_s\ket{\beta} = 0$ because $\Pi_s \ket{0} = \ket{0}$. In the second case, $\Pi_s \ket{\beta} = \Delta \ket{\beta} = \ket{\beta}$ and so $A_{G(x)} \ket{\beta} = A_{G(x)}\Delta\ket{\beta} = T(\id - \Pi_s)\ket{\beta} = T(\beta - \beta) = 0$, so by the $f(x) = 0$ case of Lemma~\ref{lem:effective-spectral-gap} with $c = 0$ we have that $\ip{0}{\beta} = 0$.
	
	The observation above implies that if we consider $\Theta < \pi$ (the Lemma is trivial otherwise), $e^{i \theta{\beta}} \not= \pm 1$ for the $\ket{\beta}$ in the support of $\ket{\hat{\zeta}}$, and so we can restrict our analysis to just the two-dimensional subspaces of $U_s$. We now split $\ip{0}{\hat{\zeta}}$:
	\begin{align*}
		\ip{0}{\hat{\zeta}} &= \bra{0}\Delta + (\id - \Delta)\ket{\hat{\zeta}} &\\
		&= \bra{0}\Delta\ket{\hat{\zeta}} + \bra{0}\Pi_s(\id - \Delta)\ket{\hat{\zeta}} & \left(\mbox{$\Pi_s\ket{0} = \ket{0}$}\right)\\
		&\leq | \bra{0}\Delta\ket{\hat{\zeta}} | + | \bra{0}\Pi_s(\id - \Delta)\ket{\hat{\zeta}} | & \left(\mbox{by the triangle inequality}\right)\\
		&\leq | \bra{0}\Delta\ket{\hat{\zeta}} | + \parnorm{\Pi_s(\id - \Delta)\ket{\hat{\zeta}}} &
	\end{align*}
	Now our goal is to bound both of the values in the last expression. First we bound $\parnorm{\Pi_s(\id - \Delta)\ket{\hat{\zeta}}}$.
	
	Given an eigenvector $\ket{\beta}$ in the support of $\ket{\hat{\zeta}}$, let $\ket{-\beta}$ be the other eigenvector in the two-dimensional subspace containing $\ket{\beta}$. Note that $\theta(\beta) = - \theta(-\beta)$. Let $\ket{\hat{\zeta}} = \sum_\beta c_\beta \ket{\beta}$, where here the sum is over all eigenvectors\footnote{Not just the ones in the support of $\ket{\hat{\zeta}}$}. Then $\parnorm{\Pi_s(\id - \Delta)\ket{\hat{\zeta}}}^2 = \parnorm{\sum_\beta \Pi_s (\id-\Delta) c_\beta \ket{\beta}}^2$. Thanks to Theorem~\ref{thm:decomposition}, we can break this summation up into pairs.
	\begin{align*}
		\parnorm{\Pi_s(\id - \Delta)\ket{\hat{\zeta}}}^2 &= \sum_{\beta: \theta(\beta) > 0} \norm{\Pi_s (\id - \Delta) (c_\beta \ket{\beta} + c_{-\beta} \ket{-\beta})}^2 &\\
		&= \sum_{S_i: \theta_i \not= 0} \norm{\frac{i}{\sqrt{2}}(c_{-\beta} - c_{\beta}) \Pi_s \ket{v_i^\bot} }^2 & \left(\mbox{rewrite $\ket{\beta}, \ket{-\beta}$ in terms of $\ket{v_i},\ket{v_i^\bot}$}\right)\\
		&= \sum_{S_i: \theta_i \not= 0} \norm{\frac{i}{\sqrt{2}}(c_{-\beta} - c_{\beta}) \sin \frac{\theta_i}{2} \ket{w_i}}^2  & \left(\mbox{change of basis}\right)\\
		&= \sum_{\beta: \theta(\beta) > 0} \left( \sin \frac{\theta(\beta)}{2} \right)^2 \norm{ \frac{i}{\sqrt{2}}(c_{-\beta} - c_{\beta}) \ket{w_i}}^2 &\\
		&\leq \sum_{\beta: \theta(\beta) > 0} \left( \sin \frac{\theta(\beta)}{2} \right)^2 \leq \left(\frac{\Theta}{2}\right)^2 & \left(\mbox{$\sin \theta \leq \theta$ for the values considered}\right)
	\end{align*}
    Next we bound the term $|\bra{0}\Delta\ket{\hat{\zeta}}$ which we will write as $|\ip{0}{w}|\parnorm{\Delta\ket{\hat{\zeta}}}$ where $\ket{w} = \Delta\ket{\hat{\zeta}}/\parnorm{\Delta\ket{\hat{\zeta}}}$ is the normalized projection of the vector $\ket{\hat{\zeta}}$ onto span of the eigenvalue zero eigenvectors of $A_G$. We will work exclusively with $\ket{w}$. First we bound the magnitude of the vector $\norm{A_{G(x)}\ket{w}}$, then decompose $\ket{w}$ into its components in the space of ``small'' and ``large'' eigenvalue eigenvectors of $A_{G(x)}$ for particular choices of ``small'' and ``large''.
    \begin{align*}
        \parnorm{A_{G(x)}\Delta\ket{\hat{\zeta}}}^2 &= \parnorm{(\id - \Pi_s)\Delta\ket{\hat{\zeta}}}^2\\
        &= \sum_{S_i: \theta_i \not= 0} \norm{\frac{i}{\sqrt{2}}(c_{-\beta} - c_{\beta}) \Delta \ket{v_i^\bot} }^2 & \left(\mbox{rewrite $\ket{\beta}, \ket{-\beta}$ in terms of $\ket{v_i},\ket{v_i^\bot}$}\right)\\
        &= \sum_{\beta: \theta(\beta) \geq 0} \left(\sin \frac{\theta(\beta)}{2}\right)^2\norm{\frac{i}{\sqrt{2}}(c_{-\beta} - c_{\beta}\ket{w_i}}^2 & \left(\mbox{change of basis}\right)\\
        &\leq \left(\frac{\Theta}{2}\right)^2\parnorm{\Delta\ket{\hat{\zeta}}}^2.
    \end{align*}
    By the definition of $\ket{w}$, we have
    \[\norm{A_{G(x)}\ket{w}}^2 = \frac{\parnorm{A_{G(x)}\Delta\ket{\hat{\zeta}}}^2}{\parnorm{\Delta\ket{\hat{\zeta}}}^2} \leq \frac{\Theta}{2}.\]
    
    For a fixed $d$, to be determined later, let $\ket{w} = \ket{w_{\mathrm{small}}} + \ket{w_{\mathrm{big}}}$ where 
    \[\ket{w_{\mathrm{small}}} = \sum_{\alpha: |\rho(\alpha)| \leq d\Theta/2} \ket{\alpha}\ip{\alpha}{w} \mbox{ and } \ket{w_{\mathrm{big}}} = \sum_{\alpha: |\rho(\alpha)| > d\Theta/2} \ket{\alpha}\ip{\alpha}{w}.\]
    Thus we have 
    \[|\bra{0}\Delta\ket{\hat{\zeta}} = |\ip{0}{w}|\norm{\Delta\ket{\hat{\zeta}}} \leq |\ip{0}{w}| \leq |\ip{0}{w_{\mathrm{small}}}| + |\ip{0}{w_{\mathrm{big}}}|\]
    where the equality is by definition, the first inequality is due to the fact that the projection of the unit vector $\ket{\hat{\zeta}}$, and the second is by triangle inequality.
    
    Bound $|\ip{0}{w}|$ as follows:
    \begin{align*}
        |\ip{0}{w_{\mathrm{small}}}|^2 &= \left(\sum_{\alpha:|\rho(\alpha)|\leq d\Theta/2} \ip{0}{\alpha}\ip{\alpha}{w}\right)^2\\
        &\leq \left(\sum_{\alpha:|\rho(\alpha)|\leq d\Theta/2} |\ip{0}{\alpha}|^2\right) \cdot \left(\sum_{\alpha:|\rho(\alpha)|\leq d\Theta/2} |\ip{\alpha}{w}|^2\right) & \left(\mbox{Cauchy-Schwartz}\right)\\
        &= \left(\sum_{\alpha:|\rho(\alpha)|\leq d\Theta/2} |\ip{0}{\alpha}|^2\right)\norm{\ket{w_{\mathrm{small}}}}^2 &\left(\mbox{definition of }\ket{w_{\mathrm{small}}}\right)\\
        &\leq 72c^2\left(1 + \frac{1}{W}\right)\norm{\ket{w_{\mathrm{small}}}}^2 &\left(\mbox{Lemma \ref{lem:effective-spectral-gap} with }c = \frac{d\Theta W}{2}\right)\\
        &\leq 6d\Theta W &\left(\mbox{$W \geq 1$ and $\ket{w}$ is normalized}\right)
    \end{align*}
    
    We further have $A_{G(x)}\ket{w} = \sum_{\alpha} \rho(\alpha)\ket{\alpha}\ip{\alpha}{w}$ so 
    \begin{align*}
        \left(\frac{\Theta}{2}\right)^2 &\geq \norm{A_{G(x)}\ket{w}}^2\\
        &= \norm{A_{G(x)}\ket{w_{\mathrm{small}}}}^2 + \norm{A_{G(x)}\ket{w_{\mathrm{big}}}}^2 &\left(\mbox{orthogonality of }\ket{\alpha}\right)\\
        &\geq d^2\left(\frac{\Theta}{2}\right)^2\norm{\ket{w_{\mathrm{big}}}}^2 
    \end{align*}
    Thus $\norm{\ket{w_{\mathrm{big}}}} \leq 1/d$. Since $\ket{0}$ is a column of the identity matrix, $\ip{0}{w_{\mathrm{big}}} \leq \norm{\ket{w_{\mathrm{big}}}}$. Together we have
    \[\sqrt{\sum_{\beta:|\theta(\beta)| \leq \Theta} |\ip{\beta}{0}|^2} = \ip{0}{\hat{\zeta}} \leq |\bra{0}\Delta\ket{\hat{\zeta}}| + \parnorm{\Pi_s(\id - \Delta)\ket{\hat{\zeta}}} \leq |\ip{0}{w_{\mathrm{small}}}| + |\ip{0}{w_{\mathrm{big}}}| + \frac{\Theta}{2} \leq 6d\Theta W + \frac{1}{d} + \frac{\Theta}{2}.\]
    Choosing $d = 1/\sqrt{6\Theta W}$, we find the bound to be $2\sqrt{6\Theta W} + \Theta/2$.  
\end{proof}

\subsection{Analysis of the Algorithms}

Given the spectral gap for $U_s$ in Lemma~\ref{lem:effective-spectral-gap-reflections}, we can analyze the algorithms. 

Algorithm~\ref{alg:phase-estimation} measures the phase of $U_s$ with input $\ket{0}$, which is in general a superposition of eigenvectors of $U_s$. If $f(s) = 1$ then by Lemma~\ref{lem:effective-spectral-gap-reflections} most of the amplitude of $\ket{0}$ is in the direction of an eigenvector with phase zero, and so the likelihood of measuring phase zero is at least $9/10$ minus the error $\delta_e$, which gives a probability of at least $4/5$. If $f(s) = 0$, then if we set $\Theta$ to be the precision $\delta_p$ then only a very small amount of the amplitude of $\ket{0}$ is in the direction of eigenvectors with phase zero: by Lemma~\ref{lem:effective-spectral-gap-reflections}, the algorithm will measure of a phase of zero with probability at most $\delta_e + (2\sqrt{6\delta_p W} + \delta_p / 2)^2 < 2/5$.

Algorithm~\ref{alg:quantum-search} prepares the state $\ket{\varphi} = \frac{1}{\sqrt{2}}(\ket{0}\ket{0} + \ket{1}U_s^T\ket{0})$ and measures the first qubit in the basis $\{\ket{+}, \ket{-}\}$, which is equivalent to measuring the first qubit of $H \ket{\varphi} = \frac{1}{2}(\ket{0}\ket{0} + \ket{1}{0} + \ket{0}U_s^T\ket{0} - \ket{1}U_x^T\ket{0})$ in the standard basis. The first qubit of $H\ket{\varphi}$ has amplitude $\frac{1}{2} + \frac{1}{2} \bra{0}U_s^T\ket{0}$ in the $\ket{0}$ direction, and so we will measure $\ket{0}$ with probability $\frac{1}{4} \parnorm{(\id + U_s^T) \ket{0}}^2$. When $f(s) = 1$, this probability will be at least $9/10$ regardless of $T$. When $f(s) = 0$,
\begin{align*}
    \Ex_{T \in [\tau]}\left[\frac{1}{4}\norm{\left(\id + U^{T}_s\right)\ket{0}}^2\right] &= \Ex_{T \in [\tau]}\left[\frac{1}{4}\sum_{\beta}|1 + \exp(i\theta(\beta)T)|^2|\ip{0}{\beta}|^2\right]\\ 
    &= \frac{1}{4}\sum_{\beta}|\ip{0}{\beta}|^2\sum_{T = 1}^{\tau}\frac{\left(2 + 2\exp(i\theta(\beta)T)\right)}{\tau}\\
    &= \frac{1}{4}\sum_{\beta}|\ip{0}{\beta}|^2\left(2 + \frac{1}{\tau}\sum_{T = 1}^{\tau}2\exp(i\theta(\beta)T)\right)\\
    &= \frac{1}{4}\sum_{\beta}|\ip{0}{\beta}|^2\left(2 + \frac{1}{\tau}\left(\sum_{T = -\tau}^{\tau}\exp(i\theta(\beta)T) -  \exp(i\theta(\beta)\cdot 0)\right)\right)\\
    &= \frac{1}{4}\sum_{\beta}|\ip{0}{\beta}|^2\left(2 + \frac{1}{\tau}\left(\frac{\exp(i\theta(\beta)(\tau + 1)) - \exp(-i\theta(\beta)\tau)}{e^{i\theta(\beta)} - 1} -  1\right)\right)\\
\end{align*}
We let $\Theta = 1/ (50W)$ and define $\nu = (2\sqrt{6\Theta W} + \Theta/2)^2$. Divide the $\ket{\beta}$ by their eigenvalues. For $\theta(\beta) \leq \Theta$, we use Lemma~\ref{lem:effective-spectral-gap-reflections} to bound the terms in the sum by $\nu$. Next consider those $\ket{\beta}$ such that $\theta(\beta) > \Theta$.
\begin{align*}
    \sum_{\ket{\beta}: \theta(\beta) > \Theta}&|\ip{0}{\beta}|^2\left(\frac{1}{2} + \frac{1}{4\tau}\left(\frac{\exp(i\theta(\beta)(\tau + 1)) - \exp(-i\theta(\beta)\tau)}{e^{i\theta(\beta)} - 1} -  1\right)\right)\\ 
    &\leq (1-\nu)\cdot \left(\frac{1}{2} + \frac{1}{4\tau}\left(\frac{\exp(i\theta(\beta)(\tau + 1)) - \exp(-i\theta(\beta)\tau)}{e^{i\theta(\beta)} - 1} -  1\right)\right)\\
    &= (1-\nu)\cdot \left(\frac{1}{2} + \frac{1}{4\tau}\left(\frac{\exp(i\Theta(\tau + 1)) - \exp(-i\Theta\tau) - \exp(i\Theta) + 1}{\exp(i\Theta) - 1}\right)\right)\\
    &= (1-\nu)\cdot\left(\frac{1}{2} + \frac{1}{4\tau}\left(\frac{\sin(\Theta(\tau + 1/2)) - \sin(\Theta/2)}{\sin(\Theta/2)}\right)\right)\\
    &= (1-\nu)\cdot \left(\frac{1}{2} + \frac{1}{4\tau\sin(\Theta/2)}\right) &\left(\mbox{$\Theta \in (0, \pi]$}\right)
\end{align*}
Thus algorithm two outputs $1$ with probability at most $\nu + (1-\nu)\cdot(1/2 + 1/(4\tau\sin(\Theta/2))$. When $\tau = \ceil{100W}$ and $W > 1$ this probability is at most $88\%$.

\section{Acknowledgements}
This survey was a project for Henry Yuen's \href{http://henryyuen.net/classes/fall2019/}{Fall 2019 course} Quantum Computing: Foundations to Frontiers. We would like to thank Gregory Rosenthal for his comments and suggestions.   

\ifdefined\AlgThree
\subsubsection{Analysis of Algorithm 3}
    The analysis of Algorithm 3 requires a bit more work so we will devote the remainder of the survey to this task. We want to show a separation between the probability that Algorithm 3 outputs 1 when $f(x) = 1$ and when $f(x) = 0$. In particular we show that the probability is greater than $64\%$ and $61\%$ in the former and latter case respectively. 
    
    Let $\tau = \ceil{10^{5}W}$. The probability that the algorithm outputs one is
    \begin{equation}
        p = \Ex_{T \in [\tau]}\left[|\bra{0}U^{T}_x\ket{0}|^2\right] = \Ex_{T \in [\tau]}\left|\sum_{\beta}e^{i\theta(\beta)T}\ip{\beta}{0}|^{2}\right|^2.
    \end{equation}\
\fi
\printbibliography

@article{DBLP:journals/jcss/Ambainis02,
  author    = {Andris Ambainis},
  title     = {Quantum Lower Bounds by Quantum Arguments},
  journal   = {J. Comput. Syst. Sci.},
  volume    = {64},
  number    = {4},
  pages     = {750--767},
  year      = {2002},
  url       = {https://doi.org/10.1006/jcss.2002.1826},
  doi       = {10.1006/jcss.2002.1826},
  bibsource = {dblp computer science bibliography, https://dblp.org}
}

@article{DBLP:journals/jcss/Ambainis06,
  author    = {Andris Ambainis},
  title     = {Polynomial degree vs. quantum query complexity},
  journal   = {J. Comput. Syst. Sci.},
  volume    = {72},
  number    = {2},
  pages     = {220--238},
  year      = {2006},
  url       = {https://doi.org/10.1016/j.jcss.2005.06.006},
  doi       = {10.1016/j.jcss.2005.06.006},
  bibsource = {dblp computer science bibliography, https://dblp.org}
}

@article{DBLP:journals/siamcomp/AmbainisCRSZ10,
  author    = {Andris Ambainis and
               Andrew M. Childs and
               Ben Reichardt and
               Robert Spalek and
               Shengyu Zhang},
  title     = {Any {AND-OR} Formula of Size {N} Can Be Evaluated in Time N\({}^{\mbox{1/2+o(1)}}\)
               on a Quantum Computer},
  journal   = {{SIAM} J. Comput.},
  volume    = {39},
  number    = {6},
  pages     = {2513--2530},
  year      = {2010},
  url       = {https://doi.org/10.1137/080712167},
  doi       = {10.1137/080712167},
  bibsource = {dblp computer science bibliography, https://dblp.org}
}

@inproceedings{DBLP:conf/coco/BarnumSS03,
  author    = {Howard Barnum and
               Michael E. Saks and
               Mario Szegedy},
  title     = {Quantum query complexity and semi-definite programming},
  booktitle = {18th Annual {IEEE} Conference on Computational Complexity (Complexity
               2003), 7-10 July 2003, Aarhus, Denmark},
  pages     = {179--193},
  year      = {2003},
  crossref  = {DBLP:conf/coco/2003},
  url       = {https://doi.org/10.1109/CCC.2003.1214419},
  doi       = {10.1109/CCC.2003.1214419},
  bibsource = {dblp computer science bibliography, https://dblp.org}
}

@article{DBLP:journals/jacm/BealsBCMW01,
  author    = {Robert Beals and
               Harry Buhrman and
               Richard Cleve and
               Michele Mosca and
               Ronald de Wolf},
  title     = {Quantum lower bounds by polynomials},
  journal   = {J. {ACM}},
  volume    = {48},
  number    = {4},
  pages     = {778--797},
  year      = {2001},
  url       = {https://doi.org/10.1145/502090.502097},
  doi       = {10.1145/502090.502097},
  bibsource = {dblp computer science bibliography, https://dblp.org}
}

@inproceedings{DBLP:conf/stoc/HoyerLS07,
  author    = {Peter H{\o}yer and
               Troy Lee and
               Robert Spalek},
  title     = {Negative weights make adversaries stronger},
  booktitle = {Proceedings of the 39th Annual {ACM} Symposium on Theory of Computing,
               San Diego, California, USA, June 11-13, 2007},
  pages     = {526--535},
  year      = {2007},
  crossref  = {DBLP:conf/stoc/2007},
  url       = {https://doi.org/10.1145/1250790.1250867},
  doi       = {10.1145/1250790.1250867},
  bibsource = {dblp computer science bibliography, https://dblp.org}
}

@article{jordan1875essai,
  title={Essai sur la g{\'e}om{\'e}trie {\`a} $ n $ dimensions},
  author={Jordan, Camille},
  journal={Bulletin de la Soci{\'e}t{\'e} math{\'e}matique de France},
  volume={3},
  pages={103--174},
  year={1875}
}

@article{DBLP:journals/siamdm/KarchmerW90,
  author    = {Mauricio Karchmer and
               Avi Wigderson},
  title     = {Monotone Circuits for Connectivity Require Super-Logarithmic Depth},
  journal   = {{SIAM} J. Discrete Math.},
  volume    = {3},
  number    = {2},
  pages     = {255--265},
  year      = {1990},
  url       = {https://doi.org/10.1137/0403021},
  doi       = {10.1137/0403021},
  bibsource = {dblp computer science bibliography, https://dblp.org}
}

@inproceedings{DBLP:conf/coco/KarchmerW93,
  author    = {Mauricio Karchmer and
               Avi Wigderson},
  title     = {On Span Programs},
  booktitle = {Proceedings of the Eigth Annual Structure in Complexity Theory Conference,
               San Diego, CA, USA, May 18-21, 1993},
  pages     = {102--111},
  year      = {1993},
  crossref  = {DBLP:conf/coco/1993},
  url       = {https://doi.org/10.1109/SCT.1993.336536},
  doi       = {10.1109/SCT.1993.336536},
  bibsource = {dblp computer science bibliography, https://dblp.org}
}

@article{DBLP:journals/cc/LaplanteLS06,
  author    = {Sophie Laplante and
               Troy Lee and
               Mario Szegedy},
  title     = {The Quantum Adversary Method and Classical Formula Size Lower Bounds},
  journal   = {Computational Complexity},
  volume    = {15},
  number    = {2},
  pages     = {163--196},
  year      = {2006},
  url       = {https://doi.org/10.1007/s00037-006-0212-7},
  doi       = {10.1007/s00037-006-0212-7},
  bibsource = {dblp computer science bibliography, https://dblp.org}
}

@article{DBLP:journals/qic/NagajWZ11,
  author    = {Daniel Nagaj and
               Pawel Wocjan and
               Yong Zhang},
  title     = {Fast amplification of {QMA}},
  journal   = {Quantum Information {\&} Computation},
  volume    = {9},
  number    = {11},
  pages     = {1053--1068},
  year      = {2009},
  url       = {http://www.rintonpress.com/xxqic9/qic-9-1112/1053-1068.pdf},
  bibsource = {dblp computer science bibliography, https://dblp.org}
}

@article{DBLP:journals/eccc/Reichardt10a,
  author    = {Ben Reichardt},
  title     = {Span programs and quantum query algorithms},
  journal   = {Electronic Colloquium on Computational Complexity {(ECCC)}},
  volume    = {17},
  pages     = {110},
  year      = {2010},
  url       = {http://eccc.hpi-web.de/report/2010/110},
  bibsource = {dblp computer science bibliography, https://dblp.org}
}

@inproceedings{DBLP:conf/soda/Reichardt11a,
  author    = {Ben Reichardt},
  title     = {Reflections for quantum query algorithms},
  booktitle = {Proceedings of the Twenty-Second Annual {ACM-SIAM} Symposium on Discrete
               Algorithms, {SODA} 2011, San Francisco, California, USA, January 23-25,
               2011},
  pages     = {560--569},
  year      = {2011},
  crossref  = {DBLP:conf/soda/2011},
  url       = {https://doi.org/10.1137/1.9781611973082.44},
  doi       = {10.1137/1.9781611973082.44},
  bibsource = {dblp computer science bibliography, https://dblp.org}
}

@article{DBLP:journals/toc/ReichardtS12,
  author    = {Ben Reichardt and
               Robert Spalek},
  title     = {Span-Program-Based Quantum Algorithm for Evaluating Formulas},
  journal   = {Theory of Computing},
  volume    = {8},
  number    = {1},
  pages     = {291--319},
  year      = {2012},
  url       = {https://doi.org/10.4086/toc.2012.v008a013},
  doi       = {10.4086/toc.2012.v008a013},
  bibsource = {dblp computer science bibliography, https://dblp.org}
}

@inproceedings{DBLP:conf/focs/Szegedy04,
  author    = {Mario Szegedy},
  title     = {Quantum Speed-Up of Markov Chain Based Algorithms},
  booktitle = {45th Symposium on Foundations of Computer Science {(FOCS} 2004), 17-19
               October 2004, Rome, Italy, Proceedings},
  pages     = {32--41},
  year      = {2004},
  crossref  = {DBLP:conf/focs/2004},
  url       = {https://doi.org/10.1109/FOCS.2004.53},
  doi       = {10.1109/FOCS.2004.53},
  bibsource = {dblp computer science bibliography, https://dblp.org}
}

@inproceedings{DBLP:conf/stoc/Tal17,
  author    = {Avishay Tal},
  title     = {Formula lower bounds via the quantum method},
  booktitle = {Proceedings of the 49th Annual {ACM} {SIGACT} Symposium on Theory
               of Computing, {STOC} 2017, Montreal, QC, Canada, June 19-23, 2017},
  pages     = {1256--1268},
  year      = {2017},
  crossref  = {DBLP:conf/stoc/2017},
  url       = {https://doi.org/10.1145/3055399.3055472},
  doi       = {10.1145/3055399.3055472},
  bibsource = {dblp computer science bibliography, https://dblp.org}
}
\appendix
\ifdefined\fullversion

\else %
\section{Span Program Example}
\label{sec:span-program-example}
For the parity function $\parity$, let $\ket{t} = [1,1]^{\top}$ and
\[A = \begin{blockarray}{ccccc}
I_{1,0} & I_{1,1} & I_{2,0} & I_{2,1}\\
\begin{block}{[ccccc]}
1 & 0 & 1 & 0\\
0 & 1 & 0 & 1\\
\end{block}
\end{blockarray}.
\]
Let all omitted $I_{i,b}$ for $b \in \{0,1\}$ implicitly set to $\emptyset$. Observe that $\wsize(P_{\parity}) = 2$. This can be achieved by a string which evaluates to false e.g. $w = 00$ with witness $\ket{y} = [0,1]^{\top}$ and $\wsize(P_{\parity}, w) = \norm{\ket{y}A}^2$ or by a string which evaluates to true e.g. $x = 01$ with witness $\ket{z} = [1,1]^{\top}$ and $\wsize(P_{\parity}, w) = \norm{\ket{z}}^2$.

The canonical span program for $\parity$ is as follows. Let the target vector be $\ket{t} = c[1,1]^{\top}$ where $c = 1/(3\sqrt{W})$. Then for $\{w_1=00, w_2=11\} = F_0$ with vector $\bra{v_{w_i,j}} \in \CC^{[m]}$ corresponding to the length $m$ vector of the $j$\textsuperscript{th} bit of $w_i$, we have
    \[A = \begin{blockarray}{ccccc}
    I_{1,0} & I_{1,1} & I_{2,0} & I_{2,1}\\
    \begin{block}{[cccc]c}
    0 & \bra{v_{w_1,1}} & 0 & \bra{v_{w_1,2}} & w_1 = 00\\
    \bra{v_{w_2,1}} & 0 & \bra{v_{w_2,2}} & 0 & w_2 = 11\\
    \end{block}
    \end{blockarray}
    \]
Further, to each string $x_i$ in $\{x_1 = 10, x_2 = 01\} = F_1$ we assign a vector $\ket{x_i} = [\ket{v_{x_i,1}}, \ket{v_{x_i,2}}]^{\top}$ where $\ket{v_{x_i, j}} \in \CC^{[m]}$ corresponds to the length $m$ vector of the $j$\textsuperscript{th} bit of $x_i$. 

Note that $m = 1$ suffices, since the matrix $A$ where
    \[A = \begin{blockarray}{ccccc}
    I_{1,0} & I_{1,1} & I_{2,0} & I_{2,1}\\
    \begin{block}{[cccc]c}
    0 & 1 & 0 & 1 & w^{(1)} = 00\\
    1 & 0 & 1 & 0 & w^{(2)} = 11\\
    \end{block}
    \end{blockarray}
    \]
and the pair of vectors $\ket{x_1} = \ket{x_2} = [1, 1]^{\top}$ satisfies the condition $\sum_{w_j \neq x_j} \ip{v_{w,j}}{v_{x,j}} = 1$.

Let us turn this canonical span program into its corresponding bipartite graphs. The matrices $B_{G(x)}$ and $B_{G'(w)}$ are then defined as follows for strings $x = 10$ and $w = 00$ which evaluates to true and false respectively on $\oplus$.
\[B_{G(x)} = \begin{blockarray}{cccccc}
\mu_0 & \BAmulticolumn{4}{c}{I}\\
\begin{block}{[c|cccc]c}
1 & 0 & 1 & 0 & 1 & \multirow{2}{*}{$F_0$}\\
1 & 1 & 0 & 1 & 0 & \\\cline{1-5}
0 & 1 & 0 & 0 & 0 & \multirow{4}{*}{$I'$}\\
0 & 0 & 0 & 0 & 0 &\\
0 & 0 & 0 & 0 & 0 &\\
0 & 0 & 0 & 0 & 1 &\\
\end{block}
\end{blockarray}\qquad
B_{G'(w)}^{*} = 
\begin{blockarray}{ccccccc}
\BAmulticolumn{2}{c}{F_{0}} & \BAmulticolumn{4}{c}{I'}\\
\begin{block}{[cc|cccc]c}
0 & 1 & 0 & 0 & 0 & 0 & \multirow{4}{*}{I}\\
1 & 0 & 0 & 1 & 0 & 0 & \\
0 & 1 & 0 & 0 & 0 & 0 & \\
1 & 0 & 0 & 0 & 0 & 1 & \\
\end{block}
\end{blockarray}.\]
These corresponds to the bipartite graphs shown Figure~\ref{fig:biadjacencymatrix-true} and Figure~\ref{fig:biadjacencymatrix-false}. 
\fig{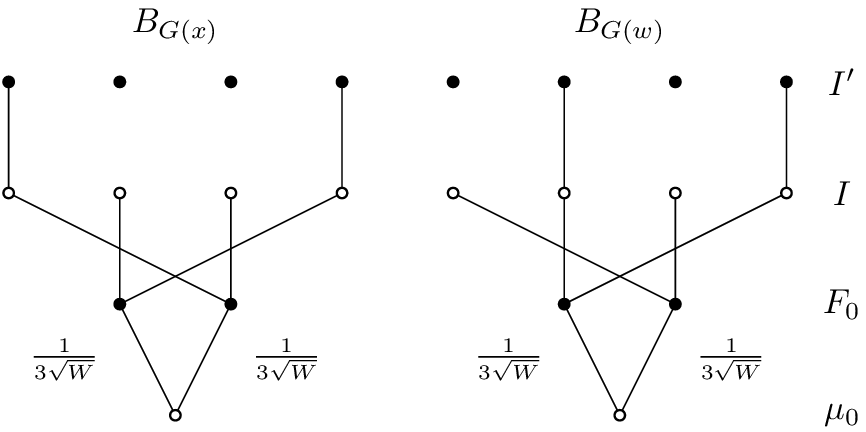}{1}{The bipartite graphs corresponding to the true biadjacency matrix. All unmarked edges have weight one. The matix vector product $B_{G(x)} \ket{\phi}$ is equivalent to assigning weights to the open dots in the picture. In order to find an eigenvalue zero eigenvector $\ket{\phi}$ of $B_{G(x)}$, the assignment of weights must ensure the neighbours of every solid dot sums to zero. Observe that $B_{G(x)}$, with $\oplus(x) = 1$, has an eigenvalue zero eigenvector while $B_{G(w)}$, with $\oplus(w) = 0$, does not. }
\fig{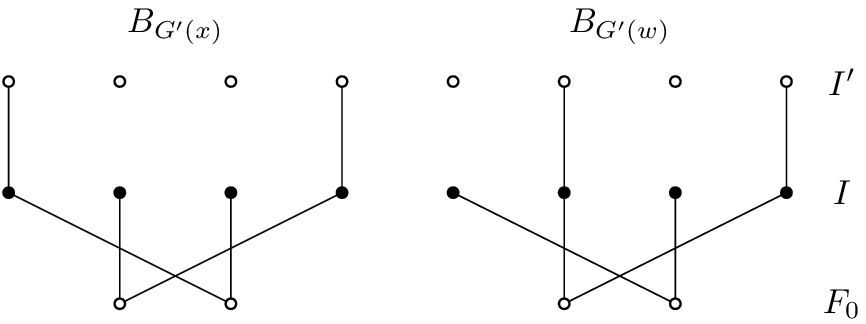}{1}{The bipartite graphs corresponding to the false biadjacency matrix. As opposed to the above, $B_{G'(x)}$, with $\oplus(x) = 1$, does not have an eigenvalue zero eigenvector while $B_{G'(w)}$, with $\oplus(w) = 0$, does. }

\section{\texorpdfstring{$\adv$}{Adv+-} is a Lower Bound for the Square Root of Formula Size}
\label{sec:adv-l-lb}

\section{Proofs for the Quantum Query Complexity Lower Bound}
\label{sec:q2-lb}

	\begin{proof}[Proof of Claim~\ref{clm:clm2}]
        
	\end{proof}

	\begin{proof}[Proof of Claim~\ref{clm:clm3}]
        
	\end{proof}
    
\section{Proof of the Decomposition Theorem}
\label{sec:decomp-proof}
    \begin{proof}[Proof of Theorem~\ref{thm:decomposition}]
        
    \end{proof}
\fi

\section{Lagrangian Duality}
\label{sec:lagrangian-duality}
Consider the following objective function: 
\begin{align*}
    \mbox{Minimize} \quad & f_0(\ket{x})\\
    \mbox{Subject to} \quad & f_i(\ket{x}) \leq 0 \mbox{ for }i \in [m]\\
    \quad & h_j(\ket{x}) = 0 \mbox{ for }j \in [p]\\
\end{align*}
for $x$ in some domain $\mathcal{D} \subset \RR^{n}$. Then the associated \emph{Lagrangian} $L: \RR^{n} \times \RR^{m} \times \RR^{j}$ is the function
\[L(\ket{x}, \ket{\lambda}, \ket{\nu}) = f_0(x) + \sum_{i = 1}^{m}\lambda_if_i(\ket{x}) + \sum_{j = 1}^{p}\nu_jh_j(\ket{x}).\]
Further, the \emph{Lagrangian dual function} is
\[g(\ket{\lambda}, \ket{\nu}) = \inf_{\ket{x} \in \mathcal{D}} F(\ket{x}, \ket{\lambda}, \ket{\nu}).\]

Observe that $g(\ket{\lambda}, \ket{\nu})$ is a lower bound for the optimal value $p^*$ of the objective function above when $\ket{\lambda} \geq 0$. Let $\ket{x}$ be any feasible solution, then $f_i(\ket{x}) \leq 0$ and $h_j(\ket{x}) = 0$. Thus
\[p^* \geq f_0(\ket{x}) \geq f_0(\ket{x}) + \sum_{i = 1}^{m}\lambda_if_i(\ket{x}) + \sum_{j = 1}^{p}\nu_jh_j(\ket{x}) = L(\ket{x}, \ket{\lambda}, \ket{\nu}) \geq \inf_{\ket{x} \in \mathcal{D}} F(\ket{x}, \ket{\lambda}, \ket{\nu}) = g(\ket{\lambda}, \ket{\nu}).\]
The best lower bound is obtained by maximizing over the dual function. In our case
\begin{align*}
    \mbox{Minimize} \qquad & t\\
    \mbox{Subject to} \qquad & \sum_{w_j \neq x_j}\bra{w,j}X\ket{x,j} = 1 \mbox{ for all } (x,w) \in \Delta, w_j \neq x_j\\
    & \sum_{j \in [n]} \bra{s,j}X\ket{s,j} \leq t \mbox{ for all } s \in \{0,1\}^n
\end{align*}
where $X \succeq 0$. The Lagrangian has one variable for every constraint. Let $Y \succeq 0$ be the variable for the constraint $X \succeq 0$, $\alpha_{w,x}$ and $\beta_{s} \geq 0$ be the variables for the equality and inequality constraints respectively. Then 
\[L = L(Y, \ket{\alpha}, \ket{\beta}; X, t) = t - \ip{X}{Y} + \sum_{(x,w) \in \Delta, w_j \neq x_j}\alpha_{x,w}\left(1 - \bra{w,j}X\ket{x,j}\right) - \sum_{\ket{s} \in \{0,1\}^n}\beta_s\left(t - \bra{s,j}X\ket{s,j}\right)\]
with dual function 
\[g(Y, \ket{\alpha}, \ket{\beta}) = \inf_{X, t} L(Y, \ket{\alpha}, \ket{\beta}; X, t).\]
Since the infimum is taken over all $X \succeq 0$ and values $t$, there exists choices of $Y$, $\ket{\alpha}$ and $\ket{\beta}$ such that $\inf_{X\succeq 0, t} L(Y, \ket{\alpha}, \ket{\beta}) = -\infty$. To remove these values from consideration, we find the implicit constraints.   

Fix $Y, \ket{\alpha}, \ket{\beta}, X$ and rewrite $L$ interms of $t$. 
\[L = t\left(1 - \sum_{\ket{s} \in \{0,1\}^n}\beta_s\right) - \ip{X}{Y} + \sum_{(x,w) \in \Delta, w_j \neq x_j}\alpha_{x,w}\left(1 - \bra{w,j}X\ket{x,j}\right) + \sum_{\ket{s} \in \{0,1\}^n}\beta_s\bra{s,j}X\ket{s,j}.\]
Since the last three terms are fixed, by taking $t \rightarrow -\infty$, $L\rightarrow -\infty$. Thus we require $1 = \sum_{\ket{s} \in \{0,1\}^n}\beta_s$. Similarly, fix $Y, \ket{\alpha}, \ket{\beta}, t$ and rewrite $L$ interms of $X$. 
\[L = \ip{X}{Z - Y} + t + \sum_{(x,w) \in \Delta}\alpha_{x,w} - t\sum_{\ket{s} \in \{0,1\}^n}\beta_s\]
where $Z = \sum_{\ket{s} \in \{0,1\}^n}\beta_s\ketbra{s}{s} - \sum_{(x,w) \in \Delta, w_j \neq x_j}\alpha_{x,w}\ketbra{w}{x}$. Again, if $\ip{X}{Z-Y} \neq 0$, then $X$ can be chosen such that $L \rightarrow -\infty$. Thus $Z = Y$. Since $Y \succeq 0$, we can simplify this to $Z \succeq 0$. 

\section{Spectral Analysis of Adjacency and Biadjacency Matrices}
\label{sec:small-eval-evectors}
Let $G$ be a weighted bipartite graph with biadjacency matrix $B_G \in \mathcal{L}(\CC^{U}, \CC^{T})$ and weighted adjacency matrix $A_{G} \in \mathcal{L}(\CC^{T \cup U})$. Further let $\ket{t} \in \CC^{T}$ and $G'$ be the graph with biadjacency matrix 
\[B_{G'} = 
    \begin{blockarray}{ccc}
        \mu_0 & U\\
        \begin{block}{[cc]c}
        \ket{t} & B_G & T\\
        \end{block}
    \end{blockarray}\]
and adjacency graph $A_{G'}$. To understand the eigenvectors of the modified adjacency matrix $A_{G'}$, we need the following theorem about eigenvectors of a PSD matrix. 

\begin{theorem}{\textup{(Spectral Bounds for PSD Matrices, Theorem 8.9 \cite{DBLP:journals/toc/ReichardtS12}.)}}
\label{thm:spectral-bound-psd-matrices}
Let $X \in \ltrans(V)$ with $X \succeq 0$, $\ket{t} \in V$, and $X' = X + \ketbra{t}{t}$. Further, let $\{\ket{\beta}\}$ be the eigenvectors of $X'$ with corresponding eigenvalue $\theta(\beta) \geq 0$. If there exists a vector $\ket{\psi}$ in the null-space of $X$ with $|\ip{t}{\psi}|^2 \geq \delta\norm{\ket{\psi}}^2$, then for any $\gamma \geq 0$
\[\sum_{\beta: |\theta(\beta)| \leq \gamma, \ip{t}{\beta} \neq 0} \frac{|\ip{t}{\beta}|^2}{\lambda(\beta)} \leq \frac{4\gamma}{\delta}.\]
\end{theorem}
Note that this sum is well defined since $\theta(\beta) \geq 0$ whenever $\ip{t}{b} \neq 0$ then $\bra{\beta}X'\ket{\beta} = \bra{\beta}X\ket{\beta} + \norm{\ip{t}{\beta}}^2 > 0$. 

\begin{theorem}{\textup{(Spectral Properties of Small Eigenvalue Eigenvectors.)}}
\label{thm:spectral-properties}
Let $G$, $B_G$, $A_G$, $G'$, $B_{G'}$, and $A_{G'}$ be as before. Suppose for some $\delta > 0$, $A_G$ has an eigenvalue zero eigenvector such that 
\[|\ip{t}{\psi_{T}}|^2 \geq \delta\norm{\ket{\psi}}^2.\]
Let $\{\ket{\alpha}\}$ be the complete set of orthonormal eigenvectors of $A_{G'}$ with corresponding eigenvalues $\rho(\alpha)$. Further, let $\ket{0}$ be the vector $[0, 1, 0]^{\top} \in \CC^{T \cup \{\mu_0\}\cup U}$. Then for all $\gamma > 0$, we have 
\[\sum_{\alpha: |\rho(\alpha)| \leq \gamma} |\ip{\alpha}{0}|^2 \leq \frac{8\gamma^2}{\delta}.\]
\end{theorem}
\begin{proof}
    The structure of the proof is as follows. We begin by reviewing relationships between the eigenvectors and eigenvalues of the adjacency graph $A_G$ and the biadjacency graph $B_G$. Given an eigenvector of $A_G$, we will relate this to the eigenvectors of the modified adjacency matrix $A_{G'}$ and modified biadjacency graph $B_{G'}$. Central to this analysis will be the study of PSD matrix $B_{G'}B_{G'}^{*}$. 

    Let $G$ be a graph and $A_G$ and $B_G$ be its adjacency and biadjacency matrices as described in the theorem statement. Let $\ket{\psi} = (\ket{\psi_{T}}, \ket{\psi_U}) \in \CC^{T\cup U}$ be an eigenvector of $A_G$ with associated eigenvalue $\rho > 0$ i.e.
    \[\begin{bmatrix}
    0 & B_G\\
    B_G^{*} & 0
    \end{bmatrix} \cdot 
    \begin{bmatrix}
    \ket{\psi_T}\\
    \ket{\psi_U}
    \end{bmatrix} = 
    \rho\begin{bmatrix}
    \ket{\psi_T}\\
    \ket{\psi_U}
    \end{bmatrix}.\]
    Then we obtain the identities $B_G\ket{\psi_U} = \rho\ket{\psi_T}$ and $B_G^*\ket{\psi_T} = \rho\ket{\psi_U}$. By negating these identities, we observe that $(\ket{\psi_T}, -\ket{\psi_U})$ is also an eigenvector of $A_G$ with associated eigenvalue $-\rho$. Observe further that $\ket{\psi_T}$, defined to be $\frac{1}{\rho}B_G\ket{\psi_U}$, is an eigenvector of $B_GB_{G}^*$ with eigenvalue $\rho^2$. Similarly $\ket{\psi_U}$, defined to be $\frac{1}{\rho}B_G^*\ket{\psi_T}$, is an eigenvector of $B_G^*B_{G}$ with eigenvalue $\rho^2$. If, instead, we begin with an eigenvector $\ket{\phi} \in \CC^{T}$ of $B_GB_G^{*}$ with eigenvalue $\lambda$, then $B_G^*\ket{\phi} \in \CC^U$ is an eigenvalue of $B_G^*B_G$ with eigenvalue $\lambda$ then 
    \[B_G^*B_G\left(B_G^*\ket{\phi}\right) = \lambda B_G^*\ket{\phi}.\]
    The pair $(\ket{\phi}, \frac{\pm 1}{\sqrt{\lambda}}B_G^*\ket{\phi})$ are eigenvectors of $A_G$ with eigenvalues $\pm\sqrt{\lambda}$ since, in the positive case for example,
    \[B_G\left(\frac{1}{\sqrt{\lambda}}B_G^*\ket{\phi}\right) = \sqrt{\lambda}\ket{\phi}\mbox{ and }B_G^*\ket{\phi} = \sqrt{\lambda}\left(\frac{1}{\sqrt{\lambda}}B_G^*\ket{\phi}\right)\]
    since $\ket{\phi}$ is an eigenvector of $B_GB_G^*$ with eigenvector $\lambda$ for the former.
    
    Let $(\ket{\psi_T}, 0)$, an eigenvector of $A_G$, be the input to our theorem. Note that $|\ip{t}{\psi_T}|^2 \geq \delta\norm{\ket{\psi}}^2$ and $B_G^*\ket{\psi_T} = 0$. We would like to bound the magnitude of
    \[\sum_{\alpha: |\rho(\alpha)|\leq \gamma} |\ip{\alpha}{0}|^2\]
    where $\{\ket{\alpha}\}$ is a complete set of orthonormal eigenvectors of $A_{G'}$ with associated eigenvalue $\rho(\alpha)$ and $\ket{0}$ is the indicator vector for entry corresponding to $\mu_0$. First we show that the eigenvalue zero eigenvectors of $A_{G'}$ are unsupported on $\mu_0$ so will not contribute to this sum. We bound $|\ip{\alpha}{0}|^2$ for eigenvectors $\ket{\alpha}$ with $0 < \rho(\alpha) \leq \gamma$ using Theorem~\ref{thm:spectral-bound-psd-matrices} by considering the eigenvectors of $B_GB_G^*$.
    
    Let $\ket{\zeta} = (\ket{\zeta_T}, \zeta_{\mu_0}, \ket{\zeta_{U}})$ be an eigenvalue zero eigenvector of $A_{G'}$. Then modified biadjacency matrix $B_{G'}$ must satisfy
    \[B_{G'}(\zeta_{\mu_0}, \ket{\zeta_{U}}) = \begin{bmatrix}
    \ket{t} & B_G
    \end{bmatrix}\cdot\begin{bmatrix}
    \zeta_{\mu_0}\\
    \ket{\zeta_{U}}
    \end{bmatrix} = \zeta_{\mu_0}\ket{t} + B_G\ket{\zeta_{U}} = 0.\]
    By multiplying both sides by $\bra{\psi_T}$, we have
    \[\zeta_{\mu_0}\ip{\psi_T}{t} + \bra{\psi_T}B_G\ket{\zeta_T} = 0,\]
    since $B_G^*\ket{\psi_T} = 0$ and $\ip{t}{\psi_T} > 0$, $\zeta_{\mu_0} = 0$. Thus eigenvalue zero eigenvectors of $A_{G'}$ are orthogonal to $\ket{0}$. 
    
    It remains to consider those eigenvectors $\ket{\alpha} = (\ket{\alpha_T}, \alpha_{\mu_0}, \ket{\alpha_U})$ of $A_{G'}$ with $\rho(\alpha) > 0$. First, using the definition of eigenvectors and the property that $A_{G'}\ket{0} = \ket{t}$, we have
    \[\rho(\alpha)\ip{\alpha}{0} = \bra{\alpha}A_{G'}\ket{0} = \ip{\alpha_T}{t}. \]
    Substituting this into our desired sum, we obtain
    \[\sum_{\alpha: 0 < |\rho(\alpha)|\leq \gamma} |\ip{\alpha}{0}|^2 = \sum_{\alpha: 0 < |\rho(\alpha)|\leq \gamma} \frac{|\ip{\alpha_{T}}{t}|^2}{\rho(\alpha)^2}.\]
    Let $B_{G'}B_{G'}^*$ be a matrix with eigenvectors $\{\ket{\beta}\}$ and corresponding eigenvalues $\theta(\beta)$. By the relationship between the eigenvalues and eigenvectors of $A_{G'}$ and $B_{G'}$ considered above, each $\ket{\beta}$ with $\theta(\beta)$ corresponds to two eigenvectors of $A_{G'}$ with eigenvalue $\left(\ket{\beta}, \frac{\pm 1}{\sqrt{\lambda(\beta})}B_{G'}^*\ket{\beta}\right)$ with $\pm\sqrt{\theta(\beta)}$. Thus
    \[\sum_{\alpha: 0 < |\rho(\alpha)|\leq \gamma} \frac{|\ip{\alpha_T}{t}|^2}{\rho(\alpha)^2} = 2\sum_{\beta: \theta(\beta) \leq \gamma^2, \theta(\beta) \neq 0} \frac{|\ip{\beta}{t}|^2}{\theta(\beta)}.\]
    Using Theorem~\ref{thm:spectral-bound-psd-matrices} with $X = B_{G'}B_{G'}^* = B_GB_G^* - \ketbra{t}{t}$ and $\ket{\psi_T}$ gives us the bound
    \[\sum_{\alpha: |\rho(\alpha)|\leq \gamma} |\ip{\alpha}{0}|^2 = 2\sum_{\beta: \theta(\beta) \leq \gamma^2, \theta(\beta) \neq 0} \frac{|\ip{\beta}{t}|^2}{\theta(\beta)} \leq \frac{8\gamma^2}{\delta}\]
    as required.
\end{proof}

Applying Theorem~\ref{thm:spectral-properties} with $\delta = 1/(9W(W+1))$ and $\gamma = c/W$ to Lemma~\ref{lem:spectra-gap} in the case where $f(s) = 0$, we obtain the following Lemma \ref{lem:effective-spectral-gap}. 
\end{document}